\documentclass[journal]{IEEEtran}
\usepackage{cite}

\ifCLASSINFOpdf
\else
\fi
\usepackage[cmex10]{amsmath}
\usepackage{algorithm}
\usepackage{algorithmic}

\usepackage{array}

\usepackage{caption}
\usepackage[font=footnotesize]{subfig}
\usepackage{url}

\usepackage{graphicx,listings,psfrag,amsfonts,cases}

\usepackage{pst-sigsys}
\usepackage{pstricks,pst-node,pst-tree}

\usepackage{amsmath,amssymb,amsfonts,graphicx,nicefrac,mathtools,bbm}
\usepackage{amsthm}

\newcommand{\BIT}{\begin{itemize}}
\newcommand{\EIT}{\end{itemize}}
\newcommand{\BNUM}{\begin{enumerate}}
\newcommand{\ENUM}{\end{enumerate}}
 
\def\reals{\mathbb{R}} 
\def\complex{\mathbb{C}} 
\renewcommand{\exp}[1]{\operatorname{exp}\left(#1\right)} 

\def\E{\mathbb{E}} 

\def\Gsn{\mathcal{N}}

\def\Unif{\textnormal{Unif}}

\newtheorem{theorem}{Theorem}
\newtheorem{lemma}{Lemma}
\newtheorem{assumption}{Assumption}

\newcommand{\putFig}[3]{
        \begin{figure}[h!] 
 		\centering
 		\includegraphics[width=#3]{figs/#1}
		  \caption{#2}
                \label{fig:#1}
        \end{figure} }

\newcommand{\sset}{\mathcal{S}}
\newcommand{\mv}{\Delta\mathbf{v}}

\usepackage{color}
\newcommand{\rev}[1]{#1}

\begin{document}
%
\title{Quick Line Outage Identification in Urban Distribution Grids via Smart Meters}


\author{Yizheng~Liao,
        Yang~Weng,~\IEEEmembership{Member,~IEEE,}
        Chin-Woo~Tan,
        Ram~Rajagopal,~\IEEEmembership{Member,~IEEE}
\thanks{Y. Liao, C-W.Tan, R. Rajagopal are with Department of Civil and Environmental
Engineering, Stanford University, Stanford, CA, 94305 USA e-mail: (\{yzliao, tancw,
ramr\}@stanford.edu). Y. Weng is with School of Electrical, Computing, and
Energy Engineering, Arizona State University, Tempe, AZ, 85287 USA e-mail:
yang.weng@asu.edu.}}


\maketitle

\begin{abstract}

%
The growing integration of distributed energy resources (DERs) in distribution grids raises various reliability issues due to DER's uncertain and complex behaviors. With a large-scale DER penetration in distribution grids, traditional outage detection methods, which rely on customers report and smart meters' ``last gasp'' signals, will have poor performance, because the renewable generators and storages and the mesh structure in urban distribution grids can continue supplying power after line outages. To address these challenges, we propose a data-driven outage monitoring approach based on the stochastic time series analysis with theoretical guarantee. Specifically, we prove via power flow analysis that the  dependency of time-series voltage measurements exhibits significant statistical changes after line outages. This makes the theory on optimal change-point detection suitable to identify line outages. However, existing change point detection methods require post-outage voltage distribution, which are unknown in distribution systems. Therefore, we design a maximum likelihood estimator to directly learn the distribution parameters from voltage data. We prove that the estimated parameters-based detection also achieves the optimal performance, making it extremely useful for fast distribution grid outage identifications. Furthermore, since smart meters have been widely installed in distribution grids and advanced infrastructure (e.g., PMU) has not widely been available, our approach only requires voltage magnitude for quick outage identification. Simulation results show highly accurate outage identification in eight distribution grids with 14 configurations with and without DERs using smart meter data.
\end{abstract}


%
\IEEEpeerreviewmaketitle

\section{Introduction}
The ongoing large-scale integration of distributed energy resources (DERs) makes photovoltaic (PV) power devices (renewable generation), energy storage devices, and electric vehicles ubiquitous. Such a change transitions the urban power grid into sustainable network and reduces the electricity cost and transmission loss \cite{sce2015application}. However, such a change also raises fundamental challenges in system operations. For example, the reverse power flow from residential houses renders the existing protective architecture inadequate. Also, frequent plug-and-charge electric vehicles will degrade power quality, causing transformer overload and voltage flickers \cite{clement2010impact}. Because of these changes on distribution grid, even a small-scale DER integration could destabilize the local grid and cause reliability issues for customers \cite{dey2010urban}. \cite{greentech} shows that the distribution power outages or blackouts caused by newly added uncertainties can cause a loss of thousands to millions of dollars within one-hour, calling for newly designed fault diagnosis approach for distribution grid operation. 

The traditional power outage analysis in distribution grids relies on passive feedback from customer reporting. Collected into Customer Information System (CIS), such information is processed in the Outage Management System (OMS) for sending field crews to identify and repair the outage. Due to the human-in-the-loop system design, delay and imprecise outage information causes inefficient detection and slow restoration. Therefore, smart meters with advanced metering infrastructure (AMI) capability were installed recently to send a ``last gasp'' message when there is a loss of power \cite{luan201417}. \cite{doe2014fault} shows additional fault location, isolation, and service restoration (FLISR) technologies to reduce some negative impact and the interruption duration.

However, the performance of the traditional methods and the recent approaches above will be degraded as the growth of DER penetration in distribution grids. For example, as shown in Fig.~\ref{fig:outage_flow}, when there is no power flow in the distribution circuit connecting to customers, the customer can still receive power from the rooftop solar panels, battery storages, and EVs. So the smart meter at the customer premises cannot report a power outage. Also, the secondary distribution grids are mesh networks in metropolitan areas \cite{rudin2012machine}, making a line outage, which may be caused by faults (e.g, short-circuit or open-circuit) and human activities, unnecessarily cause a power outage. Furthermore, some advanced secondary distribution grids have the ``self-healing'' capability, where the switches are automatically open or closed to isolate outages, restore power supply, and minimize customer impacts. However, it is still important to detect, localize, and identify the out-of-service branches for the situation awareness of distribution system operators.

\putFig{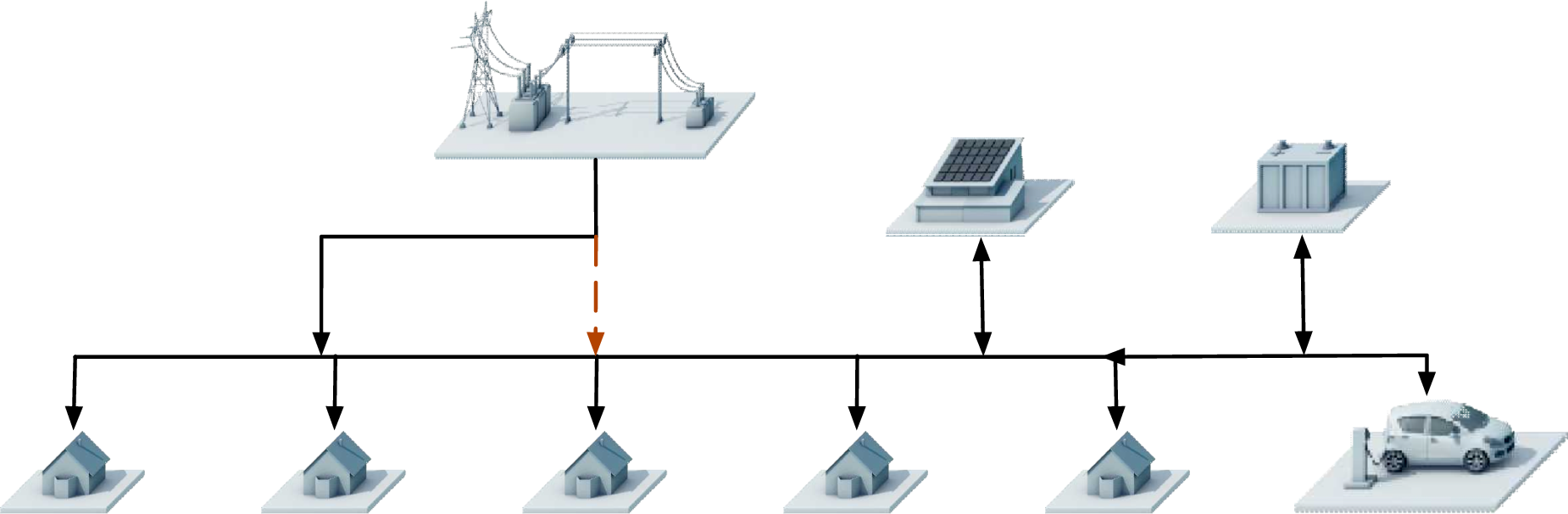}{An example of distribution grid outage. The red dashed line is the out-of-service branch.}{0.9\linewidth}

The power line outage identification in transmission grids has received a surge of interest in the past decade, where DC power flow approximation and phasor measurement units (PMUs) are the most common approaches. For example, phase changes across all buses are compared with potential fault events in \cite{tate2008line}. In \cite{he2010fault}, a transmission grid is formulated as a graphical model and phase angles are used to track the grid topology change. A regularized linear regression is employed to detect power outages in \cite{zhu2012sparse}. The approach in \cite{wei2012change} compares the branch admittance before and after outages. These methods, however, cannot be directly used in the distribution grid because $1)$ the DC approximation has poor performance in distribution grids as many systems have non-negligible line loss; $2)$ installing PMUs at all buses in distribution grid is expensive and impractical; and $3)$ the topology information is unavailable or inaccurate in distribution grids, because many DERs do not belong to the utilities and their connectivities are unknown to the system operators \cite{liao2019unbalanced}.  

For resolving the issues above, we model voltage measurement at each bus as a random variable, so that the distribution grid is modeled as a multi-variate probability distribution. We show that a line outage will lead to a change of the statistical dependence between buses' voltage data, and consequently, a change of the joint distribution. Hence, the outage can be discovered by detecting the change of the multivariate probability distribution. A well-known method to sequentially detect the probability distribution change is change point detection method, whose objective is detecting an outage as quickly as possible with a constraint of false detection rate \cite{tartakovsky2005general, liao2016urban, liao2019structural}. 

The change point detection methods have been applied to detect outage in transmission grids \cite{chen2016quickest,wei2012change,rovatsos2017statistical}. However, they cannot be directly applied because of the practical properties of distribution grids. Firstly, the outage patterns in distribution grids are usually unpredictable. With the growth of grid size, the possible post-outage distributions increase exponentially. To overcome this drawback, we propose a maximum likelihood method to directly learn the unknown post-outage probability distribution parameters from voltage data. Secondly, PMUs are not widely installed in distribution grids. Therefore, unlike the approaches in transmission grids, we cannot use the voltage phase to identify outages. We prove that voltage magnitude data, which are collected from smart meters periodically, are sufficient to detect line outages. Thirdly, the distribution grids usually have outdated or inaccurate topology \cite{liao2018urban}. Thus, precisely finding the out-of-service branch is challenging. We prove that the voltages of two disconnected buses are conditionally independent, which is subsequently used to find the line outage without knowing the post-outage probability distribution. 


The performance of our data-driven outage detection and localization algorithm is verified by simulations on the standard IEEE $8$- and $123$-bus distribution test cases \cite{kersting2001radial} and $6$ European distribution grids \cite{pretticodistribution} with $14$ network configurations. Three different real smart meter data sets are utilized for generating voltage data via data interpolation, different outage scenarios, and sensitivity analysis: Pacific Gas and Electric Company (PG\&E) data set that contains $110,000$ residential households in North California, ADRES project data set \cite{Einfalt11, VUT16} that contains 30 houses load profiles in Upper-Austria, and Pecan Street data set, which has net load data of 345 houses with root-top PV panels in Austin, Taxes.

The main contributions of this paper are summarized below:
\begin{itemize}
	\item A novel data-driven distribution grid line outage detection method is proposed. For a given probability of false alarm, the proposed outage detector is proved to have optimal detection delay.
	\item Unlike many existing works that need to know outage pattern in advance, we prove that our detection algorithm can learn the post-outage statistics directly from data. Hence, the implementation of our outage detector does not require the prior knowledge of outage pattern. Our numerical simulation demonstrates that utilizing the estimated statistics based on post-outage data does not degrade the detection performance.
	\item PMUs have not been widely installed in distribution grids. By utilizing the small angle property of distribution grids, we prove that the proposed method  only needs to use voltage magnitudes, which are usually available via smart meters, to detect line outages.
	\item We also propose an outage localization algorithm that finds the out-of-service branch after an outage event is detected. A highlight is that the proposed outage localization algorithm does not need to the distribution grid topology, which is usually required in many existing works.
	\item We validate the outage identification algorithm using three real world data sets and eight distribution grids with 17 network configurations. The numerical results illustrate that the optimality of the proposed algorithm. Additionally, multiple sensitivity analyses are conducted to show the applicability of this new line outage detection method in real world distribution grid operation.
\end{itemize}

The rest of the paper is organized as follows: Section~\ref{sec:model} introduces the modeling and the problem of the data-driven power outage detection and localization based on voltage data. Section~\ref{sec:outage_detect} uses a proof to justify that the outage can be detected by change point detection method. Also, we propose the outage detection method for only using voltage magnitudes. Section~\ref{sec:outage_identify} presents the outage localization method. A detailed algorithm for outage detection and localization is illustrated as well. Section~\ref{sec:num} evaluates the performance of the new method and Section~\ref{sec:con} concludes the paper. 

\section{System Model and Problem Formulation}
\label{sec:model}
In order to formulate the power outage detection problem, we need to describe the distribution grid and its voltage data. A distribution grid is defined as a physical network with buses and branches that connect buses. For a distribution grid with $M$ buses, we use $\mathcal{S}=\{1,2,\dots,M\}$ to represent the set of all bus indices. To utilize the time series voltage data, the voltage measurement at bus $i$ is modeled as a random variable $V_i$. We use $\mathbf{V}_\mathcal{S} = [V_1,V_2,\dots,V_M]^T$ to denote all voltage random variables in the network, where $T$ denotes the transpose operator. At the discrete time $n$, the noiseless voltage measurement at bus $i$ is $v_i[n] = |v_i[n]|\exp{j\theta_i[n]} \in \complex$, where $|v_i[n]| \in \reals$ denotes the voltage magnitude in per unit and $\theta_i[n] \in \reals$ denotes the voltage phase angle in degrees. All voltages are sinusoidal signals at the same frequency. We use $\mathbf{v}[n] = [v_1[n],v_2[n],\dots,v_M[n]]^T$ to denote a collection of all voltage measurements in a network at time $n$. Thus, $\mathbf{v}[n]$ is the realization of $\mathbf{V}_\sset$ at time $n$. Also, we use $\mathbf{v}^{1:N} = (\mathbf{v}[1],\mathbf{v}[2],\dots,\mathbf{v}[N])$ to denote a collection of all voltage  measurements in the network up to time $N$.



The problem to detect and localize line outages in a distribution grid is defined as follows:
\begin{itemize}
\item Problem: data-driven power outage detection and localization based on voltage measurements
\item Given: a sequence of the historical voltage measurements $\mathbf{v}^{1:N}$ up to the current time $N$
\item Find: (1) the outage time and (2) the branches that are out-of-service
\end{itemize}

\section{Optimal Distribution Grid Line Outage Detection}
\label{sec:outage_detect}
Voltage measurements usually have an irregular distribution and are hard to be used for our goal of this paper. Therefore, instead of using voltage measurements directly, we use the incremental change of the voltage measurements to detect outages, which is defined as $\mv[n] = \mathbf{v}[n] - \mathbf{v}[n-1]$. Accordingly, $\Delta\mathbf{v}^{1:N} = (\mv[1],\mv[2],\cdots,\mv[N])$. We use $\Delta V_i$ to represent the voltage change random variable at bus $i$ and $\Delta\mathbf{V}_\mathcal{S}$ to represent the voltage change random variables of the entire system. In the following, we will prove that, the probability distribution of $\Delta\mathbf{V}_\mathcal{S}$ will be different after an outage. In the following context, the operator $\backslash$ denotes the complement operator, i.e. $\mathcal{A}\backslash\mathcal{B} = \{i \in \mathcal{A}, i\notin \mathcal{B}\}$.

\begin{assumption}
\label{ass:indept}
	In distribution grids, 
	\begin{itemize}
		\item the incremental change of the current injection $\Delta I$ at each non-slack bus is independent, i.e., $\Delta I_i \perp \Delta I_k$ for all $i \neq k$,
		\item the incremental changes of the current injection $\Delta I$ and bus voltage $\Delta V$ follow Gaussian distribution with zero means and non-zero variances. 
	\end{itemize}
\end{assumption}

The Assumption \ref{ass:indept} has been adopted in many works, such as \cite{deka2015structure, bolognani2013identification, liao2018urban}. In \cite{liao2018urban}, the authors use real-data to validate both assumptions. According to Assumption~\ref{ass:indept}, $\Delta\mathbf{V}_\mathcal{S}$ follows a multivariate Gaussian distribution. With Assumption \ref{ass:indept}, we prove that the pairwise bus voltages are conditionally independent if there is no branch between them.

\begin{theorem}
\label{thm:cond_indept}
If the change of current injection at each bus is approximately independent and no branch connects bus $i$ and bus $k$, the voltage changes at bus $i$ and bus $k$ are conditionally independent, given the voltage changes of all other buses, i.e. $\Delta V_i \perp  \Delta V_k | \{\Delta V_e,e \in \sset \backslash\{i,k\}\}$.
\end{theorem}
\begin{proof}
For bus $i$, the current and voltage relationship can be expressed as $\Delta I_i = \Delta V_iY_{ii} - \sum_{e \in \mathcal{N}(i)}\Delta V_eY_{ie}$ with $Y_{ii} = \sum_{e \in \mathcal{N}(i)}Y_{ie}$, where $Y_{ie}$ denotes the $ie$th element of the admittance matrix $Y$ and the neighbor set $\mathcal{N}(i)$ contains the indices of the neighbors of bus $i$, i.e., $\mathcal{N}(i) = \{e \in \mathcal{S}|Y_{ie} \neq 0\}$. If bus $i$ and bus $k$ are not connected, $k \notin \mathcal{N}(i)$ and $Y_{ik} = 0$.  Given $\Delta V_e = \Delta v_e$ for all $e \in \sset\backslash \{i,k\}$, the equation above becomes to
\begin{align}
	\Delta I_i &= \Delta V_iY_{ii} - \sum_{e \in \mathcal{N}(i)}\Delta v_eY_{ie}, \nonumber \\
	\Delta V_i &= \frac{1}{Y_{ii}} (\Delta I_i + \sum_{e \in \mathcal{N}(i)}\Delta v_eY_{ie}).\label{eq:VI} 
\end{align}
Similarly, $\Delta V_k = (\Delta I_k + \sum_{e \in \mathcal{N}(k)}\Delta v_eY_{ke})/Y_{kk}$. With the assumption of the current change independence, i.e., $\Delta I_i \perp \Delta I_k$, $\Delta V_i$ and $\Delta V_k$ are conditionally independent given $\Delta \mathbf{V}_{\sset\backslash \{i,k\}}$.
\end{proof}

\putFig{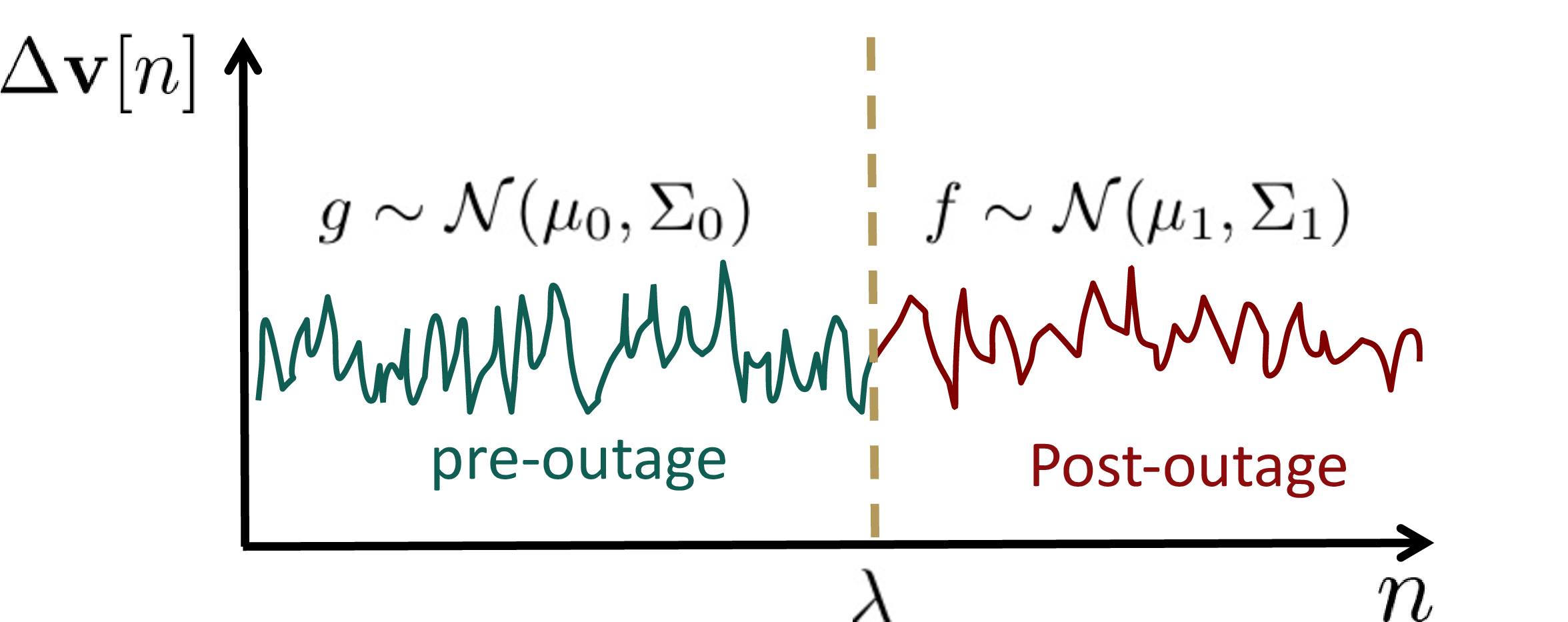}{An example of nodal voltages before and after a line outage. $\lambda$ denotes the outage occurrence time.}{\linewidth}

A branch admittance becomes zero when it is out-of-service. The voltages at the two ends of this branch become conditionally independent. Hence, the probability distribution of $\Delta\mathbf{V}_\mathcal{S}$ is different before and after an outage because some elements of the mean vector and covariance matrix will change. Let $\lambda$ denote the time that an outage occurs. We assume that $\Delta\mathbf{V}_\mathcal{S}$ follow a Gaussian distribution $g$ with the mean $\mu_0$ and the covariance matrix $\Sigma_0$ in the pre-outage status (i.e., $N \leq \lambda$) and a different Gaussian distribution $f$ with the mean $\mu_1$ and the covariance $\Sigma_1$ after any outage (i.e., $N > \lambda$). An example is illustrated in Fig.~\ref{fig:outage_plot}. One way to find the outage time $\lambda$ is performing a sequential hypothesis test at each time $N$ as follows \cite{tartakovsky2005general}:
\begin{align*}
	\mathcal{H}_0 \text{ (pre-outage)}: & \quad \lambda > N, \\
	\mathcal{H}_1 \text{ (post-outage)}: & \quad \lambda \leq N.
\end{align*}
Finding the outage time is known as the change point detection problem. Usually, the line outage occurrence time is unpredictable. Therefore, we assume the power outage time $\lambda$ as a discrete random variable with a probability mass function $\pi(\lambda)$. Now, we can use a Bayesian approach to find $\lambda$. In this paper, we assume $\lambda$ follows a geometric distribution with a parameter $\rho$. The joint distribution of $\lambda$ and $\Delta\mathbf{V}_\mathcal{S}$ can be written as
\[
P(\lambda,\Delta\mathbf{V}_\mathcal{S}) = \pi(\lambda) P(\Delta\mathbf{V}_\mathcal{S}|\lambda).
\]
When $\lambda = k$, all voltage data obtained before time $k$ follow the distribution $g$ and all the data obtained at and after time $k$ follow the distribution $f$. Therefore, the likelihood probability $P(\Delta\mathbf{V}_\mathcal{S}|\lambda)$ above is expressed as follows:
\[
P(\Delta\mathbf{V}_\mathcal{S} = \mv^{1:N}|\lambda = k) = \prod_{n=1}^{k-1}g(\mv[n])\prod_{n=k}^{N}f(\mv[n]),
\]
for $k = 1,2,\cdots,N+1$. When $\lambda = N+1$, it refers to the outage has not occurred and all data follow the distribution $g$.

Finding the outage time $\lambda$ is equivalent to finding the post-outage posterior probability $P(\mathcal{H}_1|\Delta\mathbf{V}_\mathcal{S}) = P(\lambda \leq N|\Delta\mathbf{V}_\mathcal{S} = \mv^{1:N})$ at each time $N$. If the posterior probability is large enough, we can declare an outage in the grid. At each time $N$, 
\begin{align}
	& P(\lambda \leq N|\mv^{1:N}) \nonumber \\
	=& \sum_{k=1}^N\frac{P(\lambda = k, \mv^{1:N})}{P(\mv^{1:N})}, \nonumber \\
	=& \frac{1}{P(\mv^{1:N})}\sum_{k=1}^N \pi(\lambda = k)P(\mv^{1:N}|\lambda = k), \nonumber \\
	=& C\sum_{k=1}^N\pi(k)\prod_{n=1}^{k-1}g(\mv[n])\prod_{n=k}^{N}f(\mv[n]), \label{eq:post}
\end{align}
where $C$ is a normalization factor such that $\sum_{k=1}^{N+1}P(\lambda = k|\mv^{1:N}) = 1$. In the normal operation, $f(\mv[n])$ is small and $P(\lambda \leq N|\mv^{1:N})$ is small. Once an outage occurs at time $\lambda = k \leq N$, all data collected at $n \geq \lambda$ follow $f(\mv[n])$ and $P(\lambda \leq N|\mv^{1:N})$ becomes large. Hence, we can set a threshold and declare an outage when the posterior probability surpasses this threshold. This process is visualized in Fig.~\ref{fig:detect_plot}.
\putFig{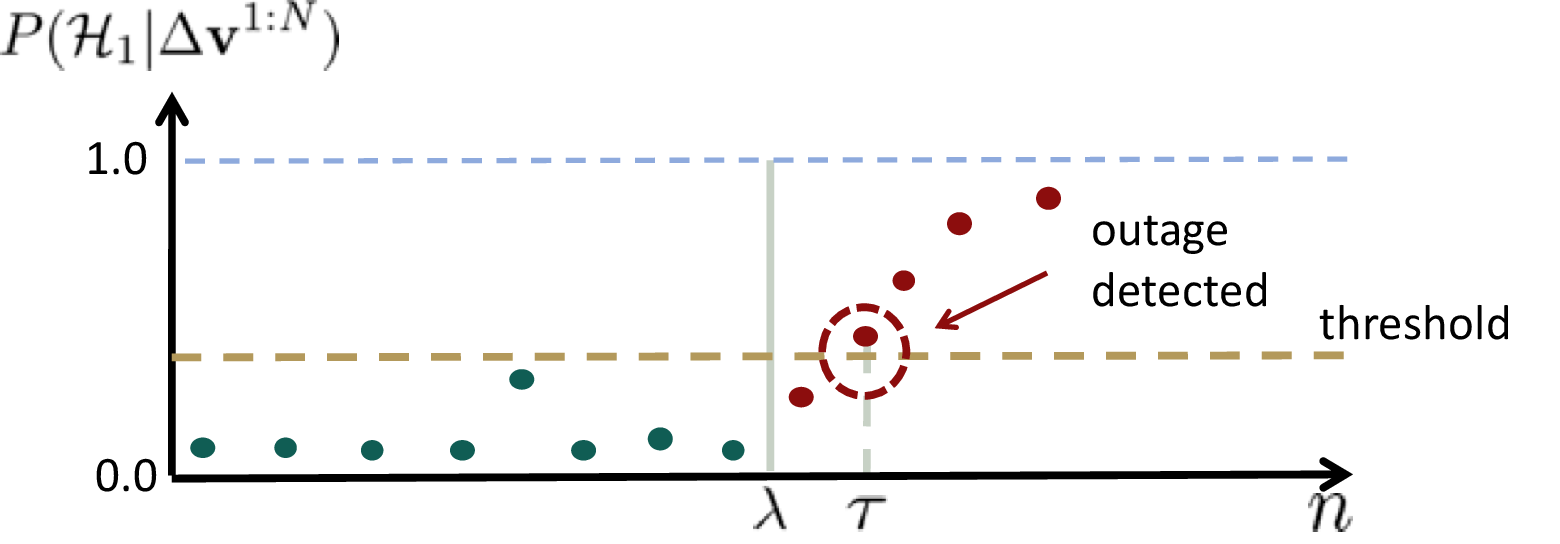}{An example of outage detection based on the posterior probability. $\lambda$ is the outage occurrence time. $\tau$ is the outage detection time. The brown dashed line is the detection threshold.}{1\linewidth}

\subsection{Optimal Outage Detection}
In the outage detection problem, we consider two performance metrics: \textit{probability of false alarm} and \textit{average detection delay}. The former metric evaluates how frequent a detector falsely declares an outage in the pre-outage status. If $\tau$ denotes the time of an outage being detected, the probability of false alarm is defined as $P(\tau < \lambda)$. The latter metric describes the average latency that a detector finds the outage after it has occurred. The average detection delay is defined as $E(\tau - \lambda | \tau \geq \lambda)$. For distribution grid line outage detection, we want to find the outage time $\lambda$ as quickly as possible with a constraint of the maximum probability of false alarm $\alpha$, i.e.,
\begin{equation}
\label{eq:obj}
\begin{aligned}
& \underset{\tau}{\text{minimize}}
& & E(\tau - \lambda | \tau \geq \lambda) \\
& \text{subject to}
& & P(\tau < \lambda ) \leq \alpha.
\end{aligned}
\end{equation}
By the Shiryaev-Roberts-Pollaks procedure \cite{pollak2009optimality}, we have the following lemma to solve the optimization problem in (\ref{eq:obj}).

\begin{lemma}
\label{thm:detect_rule}
	Given a maximum probability of false alarm $\alpha$, the following detection rule
	\begin{equation}
	\label{eq:detect_rule}
	\tau = \inf\left\{N \geq 1: P(\lambda \leq N|\mv^{1:N}) \geq 1-\alpha\right\},
\end{equation}
is asymptotically optimal \cite{tartakovsky2005general}.
\end{lemma} 
With Lemma~\ref{thm:detect_rule}, the threshold (brown dashed line) in Fig.~\ref{fig:detect_plot} is $1-\alpha$. Lemma~\ref{thm:opt_delay} shows the asymptotically optimal expected detection delay.
\begin{lemma}
	\label{thm:opt_delay}
	For a given probability of false alarm $\alpha$, the detection rule in (\ref{eq:detect_rule}) achieves the asymptotically optimal detection delay
	\rev{
	\begin{equation}
	\label{eq:bound}
	D(\tau) = E(\tau - \lambda|\tau \geq \lambda) = \frac{|\log(\alpha)|}{-\log(1-\rho)+D_\text{KL}(f\|g)},
	\end{equation}
	}
	as $\alpha \rightarrow 0$, where $D_\text{KL}(f\|g)$ is the Kullback-Leibler distance and $\log$ denotes the natural logarithm\cite{tartakovsky2008asymptotically}.
\end{lemma}

The detection process is summarized in Algorithm~\ref{alg:outage}. As a highlight, the proposed approach does not require the grid topology.

\subsection{Line Outage Detection with Unknown Outage Pattern}
Computing the posterior probability in (\ref{eq:post}) requires knowing the parameters of distributions $g$ and $f$. The parameters of pre-outage distribution $g$ can be estimated using the historical data. For obtaining the parameters of $f$, we need to know the outage pattern as a prior. One way is trying every possible outage pattern and identifying the most similar one. However, this approach is infeasible because the outage patterns can grow exponentially with the grid size. Also, many DERs in distribution grids are not operated by the utilities. Therefore, their topology information is usually unknown \cite{weng2016distributed}.

In this section, instead of searching the most likely post-outage distribution, we propose a method to learn $f$ from data using the maximum likelihood method in Lemma~\ref{lemma:mle}. The computational complexity of our approach is insensitive to the number of out-of-service branches. 
\begin{lemma}
\label{lemma:mle}
	Using observed data $\Delta\mathbf{v}^{1:N}$, The maximum likelihood estimators of the post-outage distribution $f \sim \Gsn(\mu_1, \Sigma_1)$ are
	\begin{align}
	\widehat{\mu}_1 & = \frac{\sum_{k=1}^N \pi(k)\sum_{n=k}^N \mv[n]}{\sum_{k=1}^N \pi(k) (N-k+1)}, \label{eq:mu_est} \\
	\widehat{\Sigma}_1 & = \frac{\sum_{k=1}^N \pi(k)\sum_{n=k}^N (\mv[n] - \widehat{\mu}_1)(\mv[n] - \widehat{\mu}_1)^T}{\sum_{k=1}^N \pi(k) (N-k+1)} \label{eq:sigma_est}.
	\end{align}
\end{lemma}

The proof of Lemma~\ref{lemma:mle} is given in Appendix~\ref{sec:parm_est}. With the estimates of $\mu_1$ and $\Sigma_1$, we can compute the posterior probability in (\ref{eq:post}) and apply the optimal detection rule in (\ref{eq:detect_rule}).

\subsection{Line Outage Detection with Voltage Magnitudes Only}
\label{sec:outage_voltage_mag}
Since PMUs have not been widely installed in distribution grids, the voltage phase angles are hard to be obtained in the real-world grids. To resolve this issue, in this section, we prove that the optimal line outage detection approach in Lemma~\ref{thm:detect_rule} only requires voltage magnitude data. We define the incremental change of voltage magnitude as $\Delta |v[n]| = |v[n]| - |v[n-1]|$ and use the random variable $\Delta |V|$ to represent the voltage magnitude change.

\begin{theorem}
\label{thm:cond_indept_magnitude}
If the change of current injection at each bus is approximately independent and no branch connects bus $i$ and bus $k$, the voltage magnitude changes at bus $i$ and bus $k$ are conditionally independent, given the voltage magnitude changes of all other buses, i.e. $\Delta |V_i| \perp  \Delta |V_k| \mid \{\Delta |V_e|,e \in \sset \backslash\{i,k\}\}$.
\end{theorem}
	
\begin{proof}
For bus $i$, we can rewrite (\ref{eq:VI}) as
\begin{align}
	V_i & = \frac{1}{Y_{ii}}(I_i+ \sum_{e \in \mathcal{N}(i)}V_e Y_{ie}) \nonumber \\
	V_i e^{-j\theta_i} & = \frac{1}{Y_{ii}}(I_ie^{-j\theta_i} + \sum_{e \in \mathcal{N}(i)}V_e e^{-j\theta_i} Y_{ie}) \nonumber \\
	|V_i| &=  \frac{1}{Y_{ii}}(I_ie^{-j\theta_i} + \sum_{e \in \mathcal{N}(i)}|V_e| e^{j(\theta_e - \theta_i)} Y_{ie}) \label{eq:VI_diff}.
\end{align}
In the secondary distribution grids, the phase angle difference between two neighbors' buses is relatively small \cite{kersting2012distribution}, i.e., $\theta_i - \theta_e \simeq 0$ for $e \in \mathcal{N}(i)$. Hence, (\ref{eq:VI_diff}) is approximated as
\begin{equation}
	|V_i| \simeq \frac{1}{Y_{ii}}(I_ie^{-j\theta_i} + \sum_{e \in \mathcal{N}(i)}|V_e| Y_{ie}).
\end{equation} 
For incremental change of voltage magnitude $\Delta |V_i|$, given $\Delta |V_e| = \Delta |v_e|$ for all $e \in \mathcal{S}\backslash\{i,k\}$, the equation above becomes to
\[
\Delta |V_i| = \frac{1}{Y_{ii}}(\Delta I_i e^{-j\theta_i} + \sum_{e \in \mathcal{N}(i)}\Delta |v_e| Y_{ie}).
\]
Similarly, $\Delta |V_k| = (\Delta I_k e^{-j\theta_k} + \sum_{e \in \mathcal{N}(k)}\Delta |v_e| Y_{ke})/Y_{kk}$. Since $I_i$ and $I_k$ are multiplied with constants, $\Delta I_i e^{-j\theta_i}$ and $\Delta I_k e^{-j\theta_k}$ are still independent. Hence, $\Delta |V_i|$ and $\Delta |V_k|$ are conditionally independent given $\Delta |\mathbf{V}_{\sset \backslash\{i,k\}}|$.
\end{proof}

With the proof of Theorem~\ref{thm:cond_indept_magnitude}, the optimal detection rule in (\ref{eq:detect_rule}) still holds for voltage magnitude data, i.e., 
\begin{equation}
	\tau = \inf\{ N \geq 1: P(\lambda \leq N \mid \Delta |\mathbf{v}^{1:N}|) \geq 1 - \alpha\}.
\end{equation}
For the voltage magnitude data, we can still use the maximum likelihood estimators in (\ref{eq:mu_est}) and (\ref{eq:sigma_est}) for unknown outage patterns.

\section{Out-of-Service Branch Identification}
\label{sec:outage_identify}
Identifying the out-of-service branch is important in the urban distribution grid operation. In metropolitan areas, many branches are underground and not well documented. Therefore, an efficient and accurate outage localization approach can reduce the power interruption time significantly. In the following part, we will propose a real-time outage localization method based on the voltage measurements.

\begin{lemma}
\label{thm:cond_cov}
    Assuming random vectors $\mathbf{X}$, $\mathbf{Y}$, and $\mathbf{Z}$ follow Gaussian distributions, given $\mathbf{Z} = \mathbf{z}$, if $\mathbf{X}$ and $\mathbf{Y}$ are conditionally independent, their conditional covariance is zero \cite{hastie2015statistical}. 
\end{lemma}

Because of Theorem~\ref{thm:cond_indept} and Theorem~\ref{thm:cond_indept_magnitude}, the voltage changes at the two ends of the out-of-service branches are conditionally independent after an outage. Due to Lemma~\ref{thm:cond_cov}, we can compute the conditional covariance matrix of every possible pair of buses in the grid and check if the off-diagonal term changes from a non-zero element to zero. When the off-diagonal term changes to zero, we can identify the out-of-service branches. 

Usually, the conditional covariance can be estimated based on the voltage measurements. However, a large set of post-outage data is required to have an accurate estimation, and the delay of localization is long. To enable real-time outage localization, alternatively, we use the covariance matrix $\Sigma$ to compute the conditional covariance. This approach allows us to localize the outage even if we do not know the distribution grid topology. In the case that the post-outage probability distribution $f$ is unknown, we can use $\widehat{\Sigma}_1$ in (\ref{eq:sigma_est}) to compute the conditional covariance. For bus $i$ and bus $j$, suppose $\mathcal{I} = \{i,j\}$ and $\mathcal{J} = \mathcal{S}\backslash\{i,j\}$, the covariance of the joint Gaussian distribution can be decomposed as
\[
\Sigma = \begin{bmatrix}
	\Sigma_{\mathcal{I}\mathcal{I}} & \Sigma_{\mathcal{I}\mathcal{J}} \\ 
	\Sigma^T_{\mathcal{I}\mathcal{J}} & \Sigma_{\mathcal{J}\mathcal{J}}\end{bmatrix}.
\]
The conditional covariance matrix can be computed by the Schur complement \cite{boyd2004convex}, i.e.,
\begin{equation}
	\label{eq:cond_cov}
	\Sigma_{\mathcal{I}|\mathcal{J}} = \Sigma_{\mathcal{I}\mathcal{I}} - \Sigma_{\mathcal{I}\mathcal{J}}\Sigma_{\mathcal{J}\mathcal{J}}^{-1}\Sigma^T_{\mathcal{I}\mathcal{J}}.
\end{equation}
If the voltages at bus $i$ and bus $j$ are conditionally independent, the off-diagonal term of $\Sigma_{\mathcal{I}|\mathcal{J}}$ is zero, i.e., $\Sigma_{\mathcal{I}|\mathcal{J}}(1,2) = \Sigma_{\mathcal{I}|\mathcal{J}}(2,1) = 0$. Therefore, we can compare the conditional covariance of every bus pairs before and after an outage. If the conditional covariance changes to zero after an outage, we localize one line outage event. This computation can be repeated when $\widehat{\Sigma}_1$ is updated based on the latest available measurements. In Section~\ref{sec:num}, we illustrate the similar performances using the true post-outage covariance matrix $\Sigma_1$ and the estimated covariance matrix $\widehat{\Sigma}_1$. 

\putFig{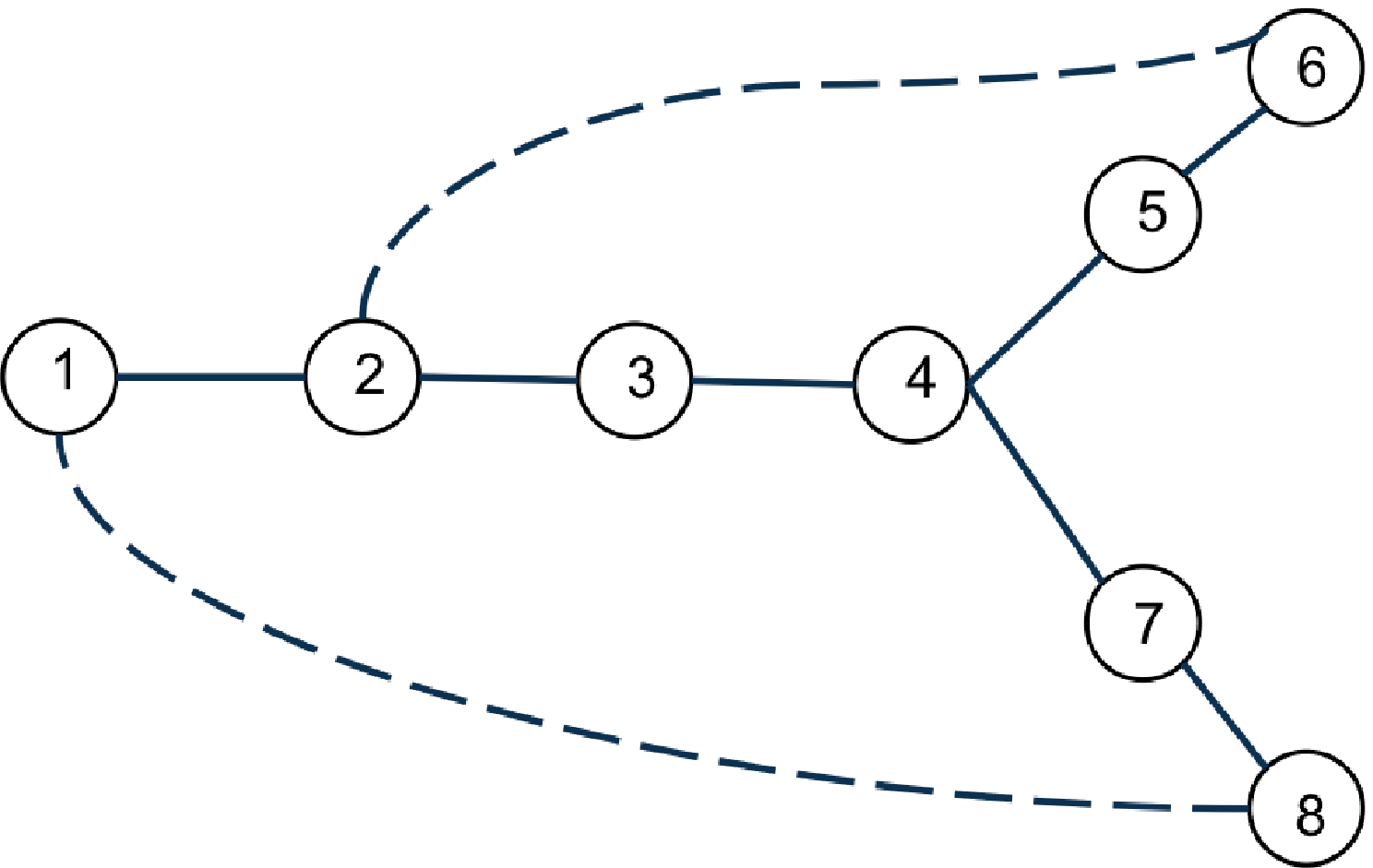}{An 8-bus system. A node represents a bus and a line represents a branch. The dashed lines are additional branches with the same admittance as the branch connected bus $7$ and bus $8$.}{0.5\linewidth}

Fig.~\ref{fig:outage_localization} visualizes the conditional correlation of a 8-bus system with loops (see Fig.~\ref{fig:8bus_loop}) before and after branch 2-6 is out-of-service. The conditional correlation between bus $i$ and bus $j$ is defined as
\begin{equation}
\label{eq:cond_corr}	
\rho_{i,j} = \frac{\Sigma_{\mathcal{I}|\mathcal{J}}(1,2)}{\sqrt{\Sigma_{\mathcal{I}|\mathcal{J}}(1,1) \times \Sigma_{\mathcal{I}|\mathcal{J}}(2,2)}}.
\end{equation}

We can observe that the conditional correlation between bus 2 and bus 6 has the most significant change. Therefore, we can locate the out-of-service branch is the branch 2-6.

\begin{figure}[htbp]
    \centering
    \subfloat[Pre-outage\label{fig:new-pre-outage}]{
    \includegraphics[width=0.47\linewidth]{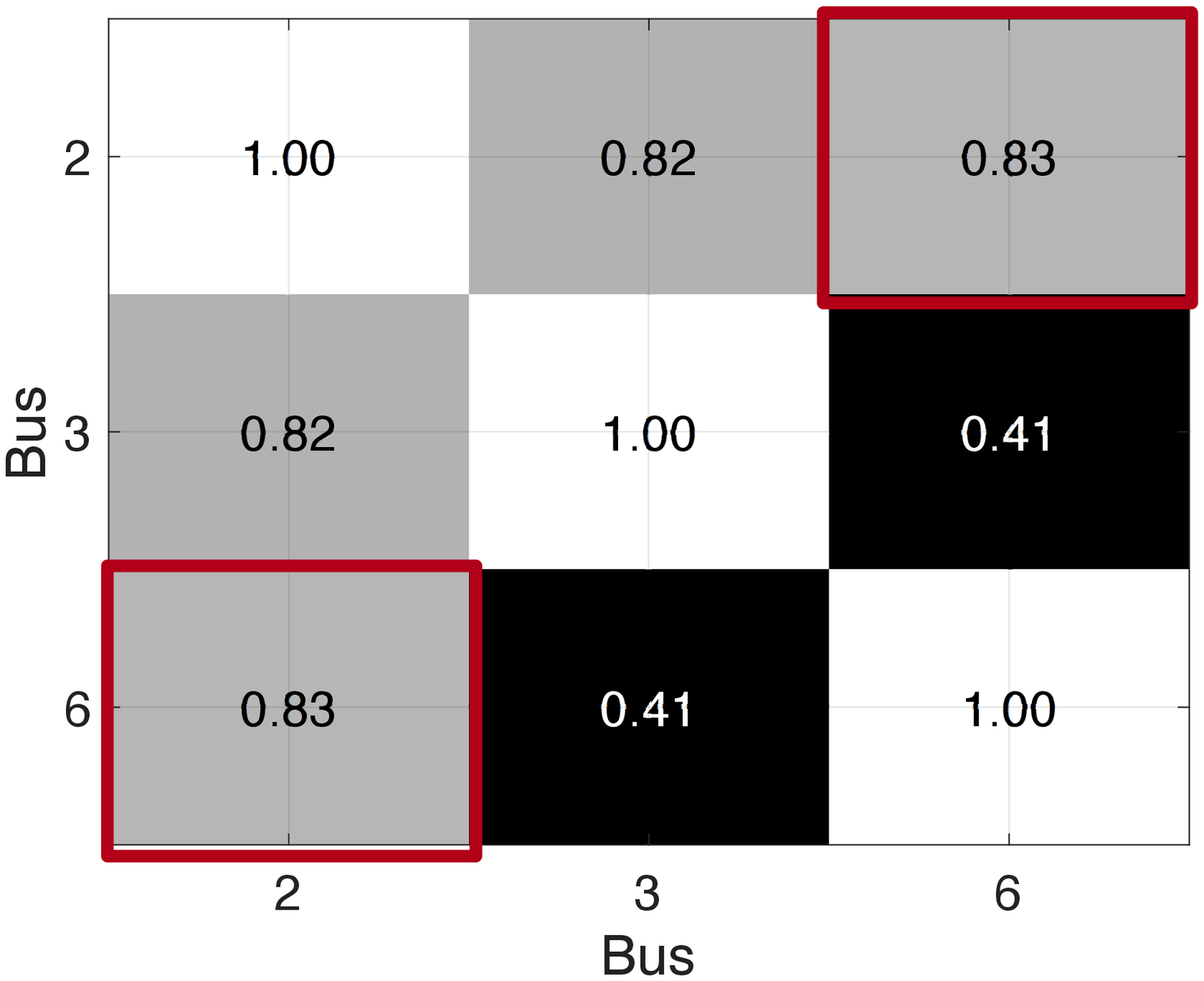}
    }
    \hfill
    \subfloat[Post-outage\label{fig:new-pre-outage}]{
    \includegraphics[width=0.47\linewidth]{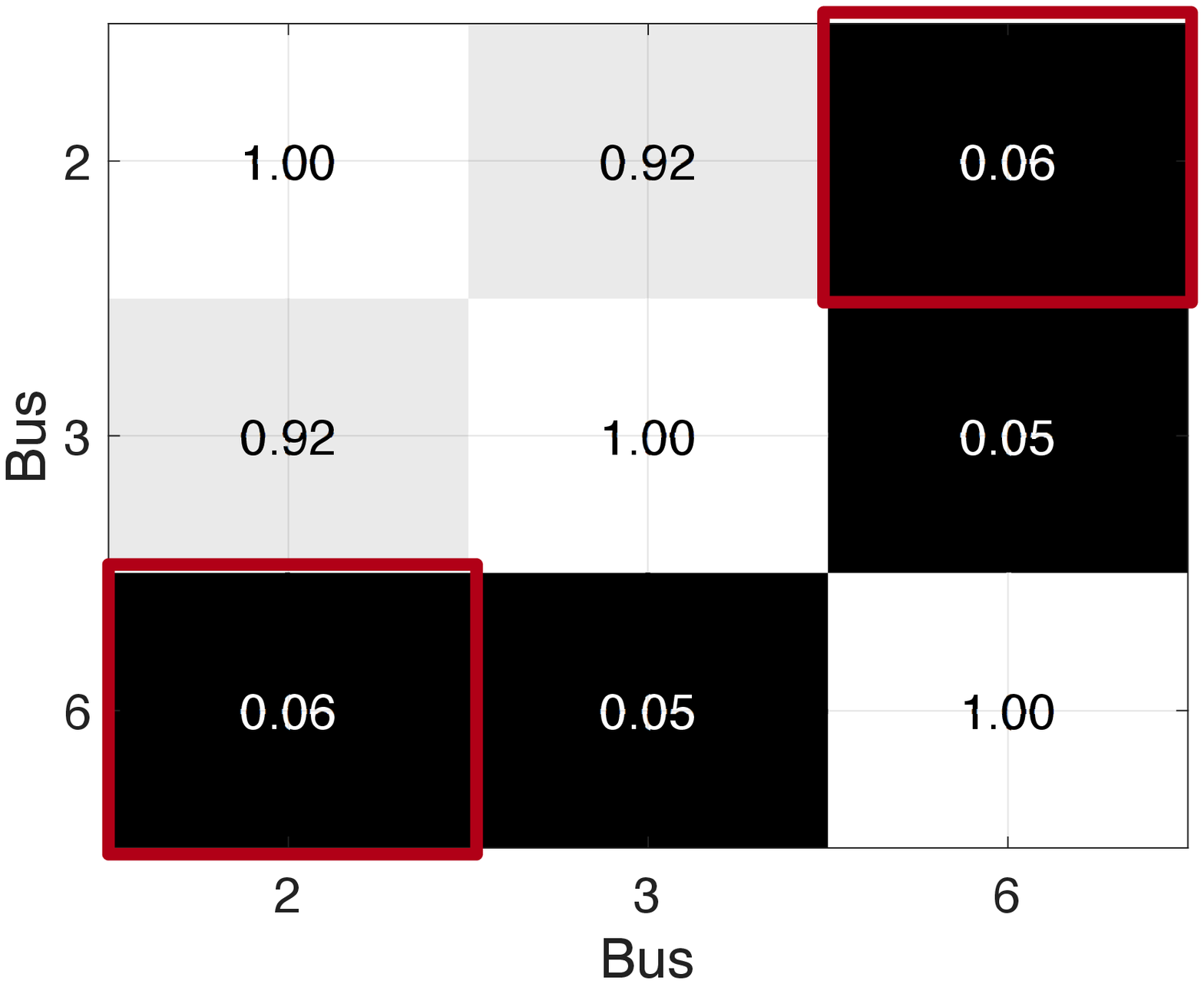}
    }
    \caption{Absolute conditional correlation before and after an outage (branch 2-6).}
    \label{fig:outage_localization}
\end{figure}

We summarize the proposed line outage detection and localization algorithm in Algorithm~\ref{alg:outage}. If only voltage magnitudes are available, we can apply the same procedure using $\Delta |\mathbf{v}_\sset^{1:N}|$.

\begin{algorithm}[h!]
\caption{Distribution Grid Line Outage Identification}
\label{alg:outage}
\begin{algorithmic}[1]
\STATE At each time $N$:
\IF {parameters of post-outage distribution $f$ are unknown}
    \STATE estimate $\widehat{\mu}_1$ and $\widehat{\Sigma}_1$ using (\ref{eq:mu_est}) and (\ref{eq:sigma_est}) with the observed data $\mv_\sset^{1:N}$
\ENDIF
\STATE Compute $P(\mathcal{H}_1|\mv_\sset^{1:N})$ by (\ref{eq:post}).
\IF {$P(\mathcal{H}_1|\mv_\sset^{1:N}) \geq 1-\alpha$}
    \STATE Report an outage event and $\tau = N$
    \STATE Compute $\Sigma_{\mathcal{I}|\mathcal{J}}$ by (\ref{eq:cond_cov}) using $\Sigma_1$ or $\widehat{\Sigma}_1$ for every pair of buses 
    \IF {$\Sigma_{\mathcal{I}|\mathcal{J}} = 0$ for $\mathcal{I} = \{i,j\}$}
        \STATE Report the branch between bus $i$ and bus $j$ is out-of-service
    \ENDIF
\ENDIF
\end{algorithmic}
\end{algorithm}

\rev{At time $N$, the computational complexity of outage detection only depends on the grid size, as shown in (\ref{eq:post}) and (\ref{eq:detect_rule}). As presented in (\ref{eq:cond_cov}) and (\ref{eq:cond_corr}), the computational complexity of outage localization also only depends on the grid size.} In our numerical simulations, for distribution grids with up to $200$ buses, the process outlined in Algorithm~\ref{alg:outage} can be completed within 10 seconds using a modern desktop computer at each time $N$. Compared with smart meter sampling rate, which is usually ranging from 1 minute to 1 hour, the computational delay of the outage identification is negligible. Hence, the proposed line outage identification algorithm can be used for real-time applications.  

    
    



\section{Simulation and Results}\label{sec:num}
The simulations are implemented on the IEEE PES distribution networks for IEEE $8$-bus and $123$-bus networks \cite{kersting2001radial} and six European distribution grids \cite{pretticodistribution}. To validate the performance of the proposed approach on loopy networks, we add several branches to create loops in all systems. The loopy $8$-bus system is shown in Fig.~\ref{fig:8bus_loop}. For $123$-bus system, we add a branch between bus 77 and bus 120 and the other branch between bus 50 and bus 56. The admittance are the same as the branch between bus 122 and bus 123. For European systems, the loopy modifications are detailed in \cite{liao2018urban}. In each network, bus $1$ is selected as the slack bus. The historical data have been preprocessed by the MATLAB Power System Simulation Package (MATPOWER) \cite{Zimmerman10}. 

We use the real power profile of distribution grids from Pacific Gas and Electric Company (PG\&E) in the subsequent simulation. This profile contains anonymized and secure smart meter readings over $110,000$ PG\&E residential customers for one year spanning from $2011$ to $2012$. The reactive power $q_i[n]$ at bus $i$ and time $n$ is computed according to a randomly generated power factor $pf_i[n]$, which follows a uniform distribution, e.g. $pf_i[n] \sim \Unif(0.8,1)$. 
To obtain measurements form voltage phasors at time $n$, i.e. $v_i[n]$, we run a power flow to generate the states of the power system. To obtain time-series data, we run the power flow to generate voltage data over a year. 


In this simulation, we considered three common outage scenarios:
\begin{enumerate}
    \item Mesh networks. In this system, after an outage, most buses will not have zero voltages because they can receive powers from multiple branches. This outage scenario usually happens in urban areas.
    \item Radial networks with high DER penetrations. In this case, some buses will be disconnected from the main grid. However, they are still powered by DERs and thus, their voltages will not be zero. This outage case is a typical scenario in residential areas. 
    \item Radial networks without DERs. In this case, when a line outage occurs, some buses will be disconnected from the main grid and have zero voltage magnitudes. These smart meters stop to transmit measurements, but they send last gasp messages before disconnecting from the grids. Therefore, we can set measurements from all smart meters that send last gasp message zero. Because the bus voltages have no variation after outages, our method can quickly detect and localize this type of outages.
\end{enumerate}

When multiple induction motors are presented in distribution grids,  residual voltages may exist after the terminal buses disconnect from the main grid \cite{akiyama1990induction}. If the residual voltage is above smart meter measurement threshold and lasts for a certain period of time (e.g., a few minutes to an hour, depending on the smart meter sampling frequency), the outage case is similar to outage scenario 2 above. If the residual voltage is below the measurement threshold, smart meters may not report  measurements. In this case, smart meters send last gasp signal and the outage detection case is similar to the outage scenario 3 above.

\subsection{Outage Detection in Mesh Distribution Grids}
\label{sec:mesh_sim}
Fig.~\ref{fig:8bus_loop_det} illustrates the complementary posterior probability $1-P(\mathcal{H}_1|\mv^{1:N})$ for detecting two line outages in loopy 8-bus system (Fig.~\ref{fig:8bus_loop}) based on voltage magnitude data $\Delta |\mathbf{V}_\mathcal{S}|$. In this test, branches 3-4 and 2-6 have outages. The false alarm rate is $10^{-6}$. For the complementary posterior probability, the threshold is $\alpha = 10^{-6}$.
To have a better understanding of how our proposed outage detection algorithm works, we assign a uninformative parameter for the prior distribution, i.e., $\rho = 10^{-4}$. The outage time is $\lambda = 21$. When the parameters of post-outage distribution are known, the complementary posterior probability immediately drops below the threshold at $N=21$. When the parameters are unknown, one more time step is required to achieve detectable probability. Since the voltage magnitudes are collected every hour, the additional delay is one hour when the outage pattern is unknown. We want to highlight that although the delay is one hour, the customers do not experience power outage because of the mesh structure. Later, we show that we can reduce the latency by increasing the sampling frequency of smart meters.

\putFig{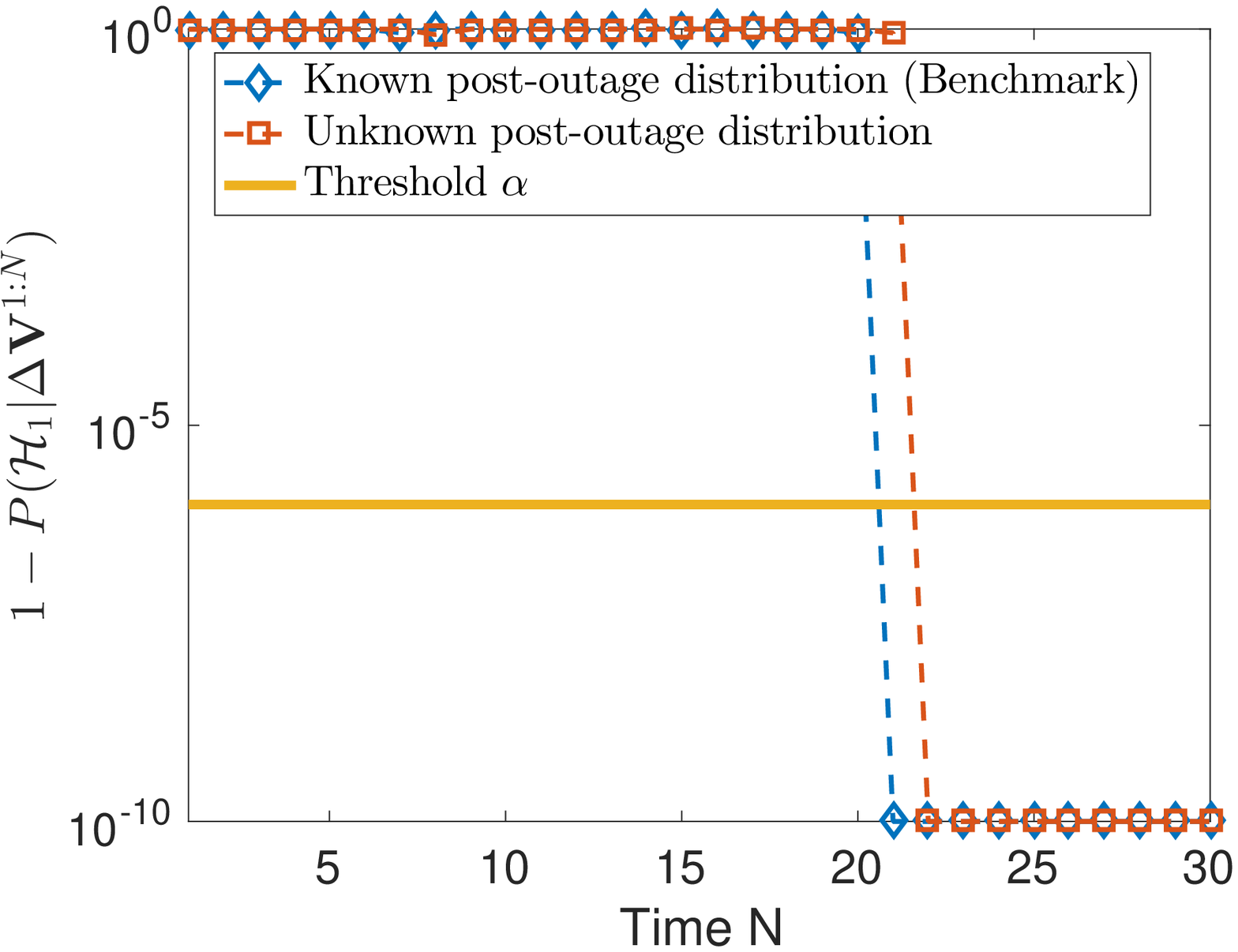}{Complementary posterior probability for outage detection. The branches 3-4 and 2-6 have outage. $\alpha = 10^{-6},\rho = 10^{-4}$.}{\linewidth}

In Fig.~\ref{fig:123busloop_delay}, the expected delay divided by \rev{$|\log(\alpha)|$} is plotted as a function of \rev{$|\log(\alpha)|$} for two cases: $f$ is known and $f$ is unknown. The choices of abscissa and ordinate are motivated by Lemma~\ref{thm:opt_delay}. Specifically, the asymptotically optimal detection delay in Lemma~\ref{thm:opt_delay} can be rewritten as
\rev{
\[
\frac{D(\tau)}{|\log(\alpha)|} = \frac{1}{-\log(1-\rho)+D_\text{KL}(f\|g)}.
\]
}
For a particular outage pattern, the KL distance between the pre-outage distribution $g$ and the post-outage distribution $f$ is fixed. Additionally, if the prior distribution is known, $-\log(1-\rho)+D_\text{KL}(f\|g)$ is a constant. Hence, the detection delay $D(\tau)$ becomes a function of probability of false alarm $\alpha$. Plotting the relationship between \rev{$|\log(\alpha)|$} and \rev{$\frac{D(\tau)}{|\log(\alpha)|}$} helps to explore the asymptotical property of the proposed algorithm. We also show the limiting value of the normalized asymptotically optimal detection delay $1/(-\log(1-\rho)+D_\text{KL}(f\|g))$ in Fig.~\ref{fig:123busloop_delay}. All plots are generated by Monte Carlo simulation over $1,000$ replications. In this simulation, the prior distribution of outage time $\lambda$ has a geometric probability distribution with parameter $\rho = 0.04$. The start time of test is randomly selected within one year. In Fig.~\ref{fig:123busloop_delay}, our approach, which learns the parameters of the post-outage distribution from the voltage measurements, has identical performances as the optimal method that has known $f$. Also, our approach can achieve the optimal expected detection delay asymptotically. As shown in Fig.~\ref{fig:123busloop_delay}, when the false alarm rate $\alpha$ is small, our approach can report the outage immediately (i.e., detection delay is less than one hour), which can significantly reduce the impacts of power outages. In \cite{banerjee2014power}, an optimal change-point detection approach is proposed to identify line outages in transmission grids using PMU data. Although the grid type is different, our method has the similar performance as \cite{banerjee2014power} and both converge to the asymptotical detection delay bound $1/(-\log(1-\rho)+D_\text{KL}(f\|g))$. \rev{Specifically, for the loopy 123-bus system, with $\alpha=10^{-5}$, our algorithm needs $4.89$ time steps to detect outages by using $\Delta|\mathbf{V}_\mathcal{S}|$. The algorithm in \cite{banerjee2014power} uses $\Delta\mathbf{V}_\mathcal{S}$ for outage detection and requires $4.91$ time steps to detect outages.} Hence, both methods need the same amount of data for detecting outages but our method only requires the smart meter data.
 
\putFig{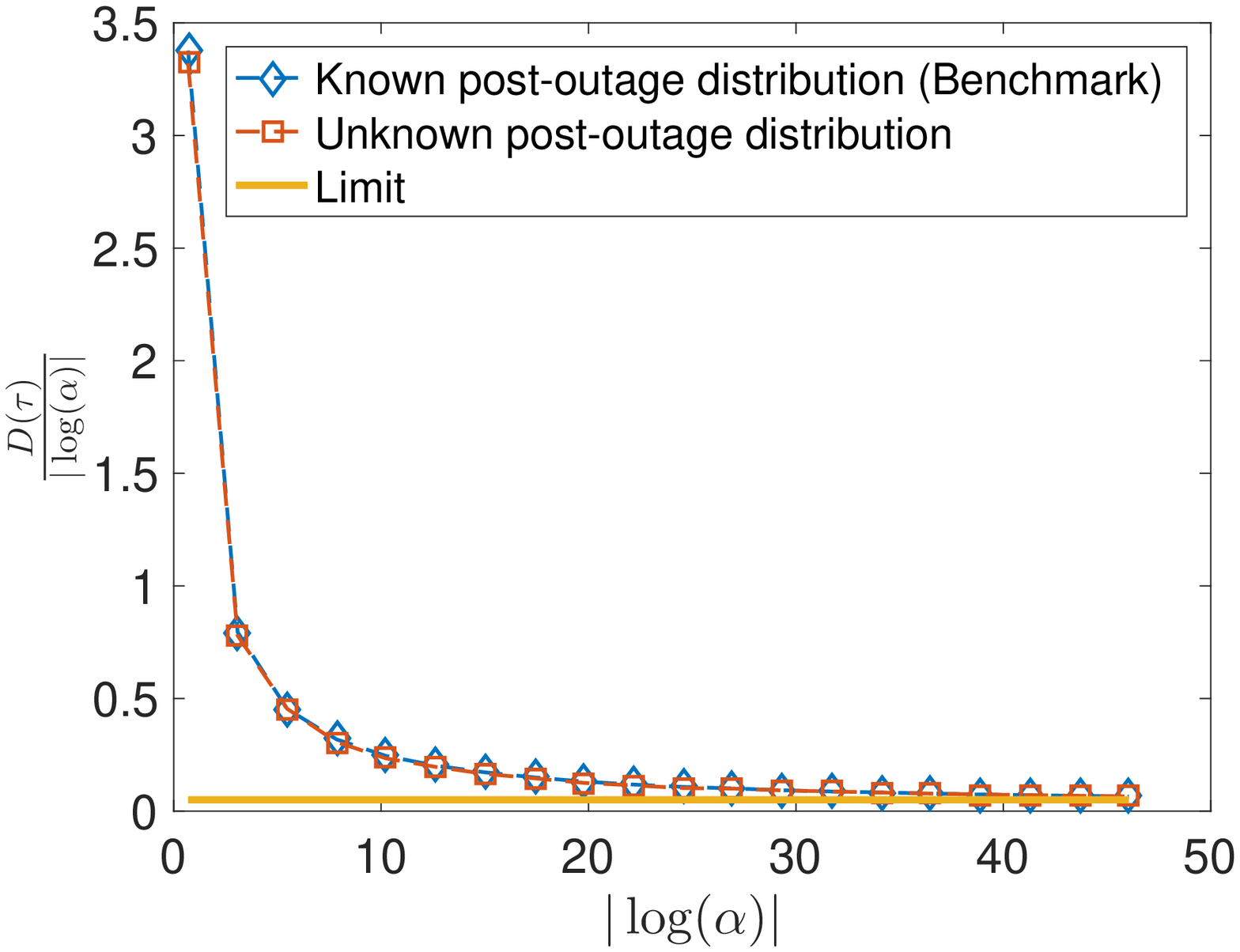}{\rev{Plots of the slope $\frac{D(\tau)}{|\log(\alpha)|}$ against $|\log(\alpha)|$ for outage detection for loopy $123$-bus system.} False alarm rate $\alpha$ ranges in $[0.5,10^{-20}]$. Branch 73-74 has an outage.}{\linewidth}

\subsection{Outage Detection in Radial Distribution Grids with DERs}
\label{sec:radial_sim}
In a radial distribution grid, a line outage will lead to several isolated islands. However, with the integration of DERs, such as solar panels and batteries, some buses can still receive powers. In mesh systems, the continuous power supply from DERs also makes the outage detection difficult. In this section, we simulate the line outage in IEEE 8-bus and 123-bus systems and six European medium- and low-voltage distribution systems based on voltage magnitude data \cite{pretticodistribution, liao2018urban}. Similar to the previous section, we randomly select the start time within one year. Also, we select a few buses in the distribution grid to have solar power generator with a battery as the storage. Thus, there is a power supply during the entire day. If the battery is unavailable, the outage can be directly detected when the nodal voltages are zero. For the solar panel, we use the power generation profile computed by PVWatts Calculator, an online application developed by the National Renewable Energy Laboratory (NREL) \cite{dobos2014pvwatts}. The solar power generation profile is computed based on the weather history in North California and the physical parameters of ten 5kW solar panels. The power factor is fixed as $0.90$ lagging, which satisfies the regulation of many U.S. utilities \cite{ellis2012review} and the IEEE standard \cite{ieee2014guide}.

\begin{table}[h!]
\caption{Average Detection Delay (Time Step) of Line Outage Detection in Distribution Grids with DERs. $\alpha = 10^{-5}$. The post-outage Distribution $f$ is Unknown.}
\centering
	\begin{tabular}{|c||c|c|c|c|}
	\hline
	System & Total & Total & $\Delta \mathbf{V}_\sset$ & $\Delta |\mathbf{V}_\sset|$ \\
	& Branches & DER & (1 min) & (60 min) \\
	\hline
	8-bus & 7 & 8 & 0.12 & 0.12 \\
	\hline
	8-bus, 2 loops & 9 & 8 & 0.13 & 0.15 \\
	\hline
	123-bus & 122 & 12 & 3.62 & 4.77 \\
	\hline
	123-bus, 2 loops & 124 & 12 & 3.53 & 4.89 \\
	\hline
	\textit{LV\_suburban} & 114 & 10 & 2.81 & 5.00 \\
	\hline
	\textit{LV\_suburban} & 114 & 20 & 2.99 & 5.00 \\
	\hline
	\textit{LV\_suburban} & 114 & 33 & 3.23 & 5.00 \\
	\hline
	\textit{LV\_suburban\_mesh} & 129 & 33 & 4.95 & 5.83 \\
	15 loops &&&& \\
	\hline
	\textit{MV\_urban} & 34 & 7 & 1.11 & 2.02  \\
	\hline
	\textit{MV\_urban} & 35 & 7 & 1.11 & 1.29 \\
	switch 34-35, 1 loop &&&& \\ 
	\hline
	\textit{MV\_urban} & 37 & 7  & 1.12 & 1.29 \\
	3 switches, 3 loops &&&& \\ 
	\hline
	\textit{MV\_two\_stations} & 46 & 10 & 0.92 & 1.33 \\
	\hline
	\textit{MV\_two\_stations} & 48 & 10 & 0.87 & 1.35 \\
	2 switches, 2 loops &&&& \\
	\hline
	\textit{MV\_rural} & 116 & 20 & 1.13 & 2.44 \\
	\hline
	\textit{MV\_rural} & 119 & 20 & 1.98 & 3.01 \\
	3 switches, 3 loops &&&& \\
	\hline
	\textit{Urban} & 3237 & 300 & 11.89 & 29.23 \\
	\hline
	\textit{LV\_large}, 465 loops & 4030 & 300 & 33.29 & 88.40 \\
	\hline
	\end{tabular}
\label{tab:detection_DER}
\end{table}

Table~\ref{tab:detection_DER} summarizes the average detection delay in eight distribution grids with $14$ configurations. In each network, we compare the detection performance between voltage magnitude and phase ($\Delta \mathbf{V}_\sset$) and voltage magnitude only ($\Delta |\mathbf{V}_\sset|$). We choose $\Delta \mathbf{V}_\sset$ with $1$ minute sampling rate to demonstrate the relative faster metering speed and compare to $\Delta \mathbf{V}_\sset$ with $1$ hour for normal smart meters data. We use linear interpolation method to generate the $1$ minute data from the hourly power profile. Although the sampling frequencies are different, the additional amount of voltage magnitude data for outage detection is relatively small (1-3 time steps) for most networks. This highlights that using voltage magnitude can achieve the similar detection performance as using both voltage magnitude and phase angles. Compared with the distribution grid line outage identification method proposed in \cite{sevlian2017outage}, our approach needs fewer samples with the same probability of false alarm. \rev{For the IEEE 123-bus system, which is a radial network, our algorithm has a detection delay of $4.77$ time steps using $\Delta |\mathbf{V}_\mathcal{S}|$ with $\alpha=10^{-5}$ and the method in \cite{sevlian2017outage} has a delay of $10.45$ time steps with the same $\alpha$. A note is that we do not optimize the sensor placement for the approach in \cite{sevlian2017outage}, which may reduce the detection delay.} Also, the method in \cite{sevlian2017outage} can only be applied to radial networks but ours can be deployed to both radial and mesh grids.

For large-scale distribution grids, we need more data to detect outages when only voltage magnitudes are available. The reason is that the dimension of the covariance matrix is high and more data are needed for accurate estimation. When some grid topology information is known, this issue can be addressed by decomposing the covariance matrix since the distribution grid is usually sparse. For example, in the MV distribution grid presented in Fig.~\ref{fig:MV_rural}, there are multiple LV distribution grids and each of them is connected via a common MV grid. Therefore, we only need to identify outage within each LV grid and apply another outage detector for the MV grid. Such way can help to reduce the computational complexity. Another case is that we can split the grid into different clusters when a detailed topology is available. In Fig~\ref{fig:LV_semiurban}, we can apply outage identification to each clusters for buses and, hence, reduce the dimension of covariance matrix.

\putFig{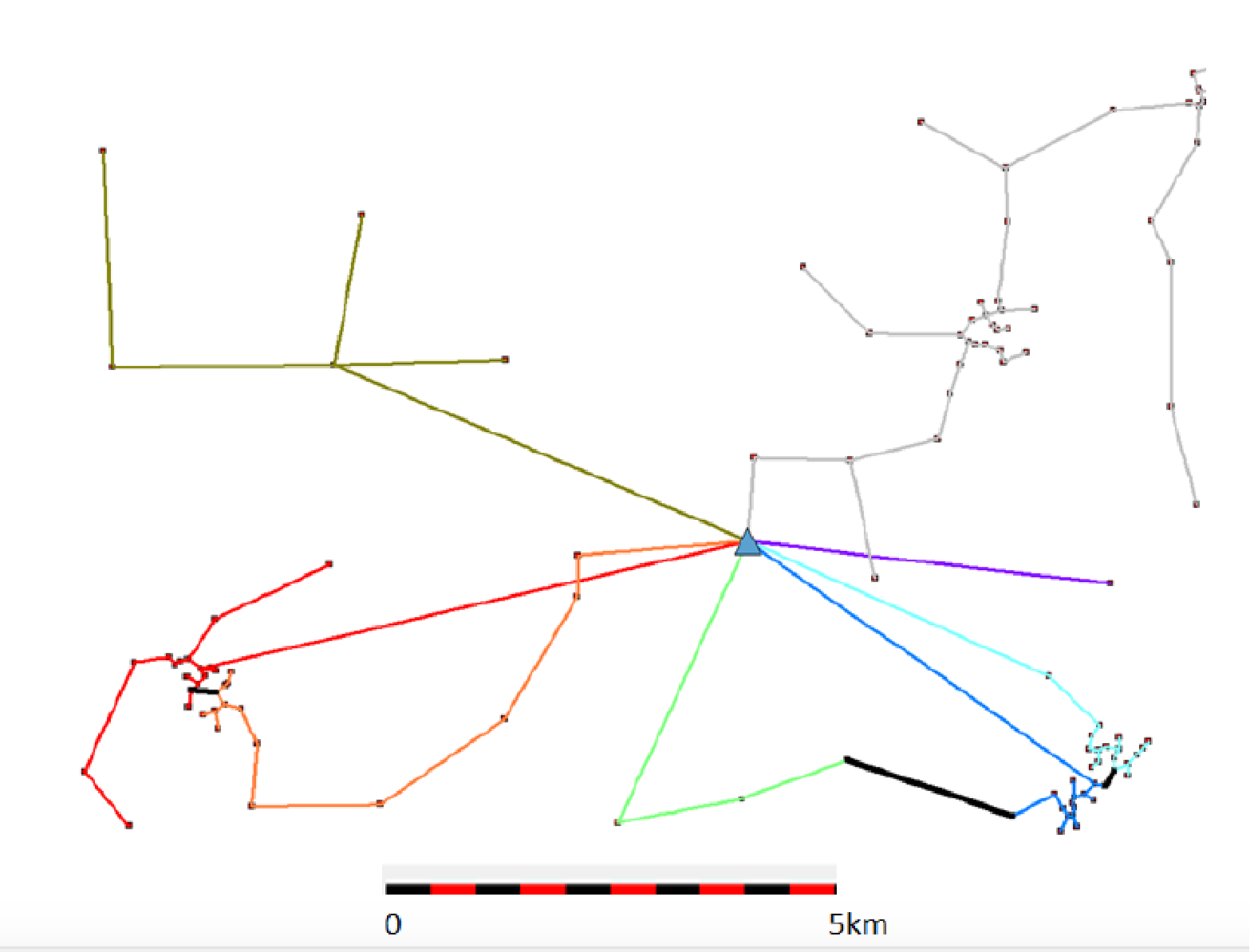}{A rural median voltage distribution grid (MV\_rural) \cite{pretticodistribution}. Each color represents one medium voltage branch. The dots represent the substations and the low voltage grids are connected via these substations.}{0.8\linewidth}
\putFig{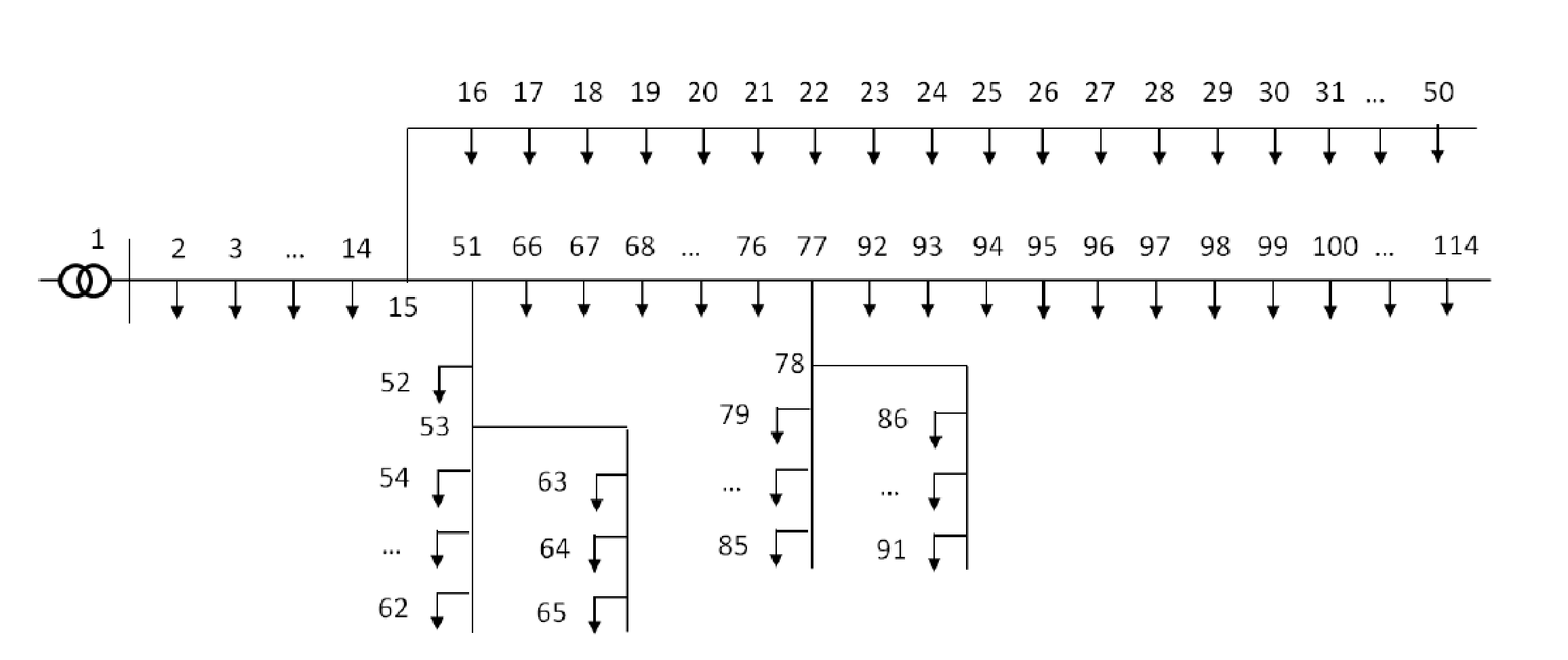}{A suburban low voltage network (LV\_suburban)\cite{pretticodistribution}.}{\linewidth}

The roof-top solar power generation can highly correlated within one LV distribution grid and may jeopardize Assumption~\ref{ass:indept}. To validate our algorithm in this scenario, we use the data from Pecan Street \cite{nagasawa2012data}, which contains hourly load measurements for 345 houses with roof-top PV integrations in Austin, Taxes. The measurements include both power consumption and renewable generation. Table~\ref{tab:pecan} summarizes the average detection delay using Pecan Street data for both radial and mesh distribution grids. Compared wit the results of the same grid in Table~\ref{tab:detection_DER}, we do not observe any major performance degradation. Hence, the results in Table~\ref{tab:pecan} demonstrate that our proposed algorithm can be applied to high-penetration grids. 

\begin{table}[h!]
\caption{Average Detection Delay (Time Step) of Line Outage Detection using Pecan Street Data. $\alpha = 10^{-5}$. The post-outage Distribution $f$ is Unknown.}
\centering
	\begin{tabular}{|c||c|c|c|c|}
	\hline
	System & Total & Total & $\Delta \mathbf{V}_\sset$ & $\Delta |\mathbf{V}_\sset|$ \\
	& Branches & DER & (1 min) & (60 min) \\
	\hline
	123-bus & 122 & 12 & 2.91 & 5.05 \\
	\hline
	123-bus, 2 loops & 124 & 12 & 3.37 & 4.59 \\
	\hline
	\textit{LV\_suburban\_mesh} & 129 & 33 & 5.08 & 5.70 \\
	15 loops &&&& \\
	\hline
	\textit{MV\_urban} & 34 & 7 & 1.28 & 3.11  \\
	\hline
	\end{tabular}
\label{tab:pecan}
\end{table}

\subsection{Line Outage Localization}
When a branch has an outage, the conditional correlation defined in (\ref{eq:cond_corr}) becomes zero. Fig.~\ref{fig:outage_line} shows the absolute conditional correlation $|\rho_{i,j}|$ of the loopy $8$-bus system in Fig.~\ref{fig:8bus_loop} after branch 3-4 and branch 2-6 have outages. The red boxes indicate the branches that have outages. When the post-outage distribution $f$ is known, the true $\Sigma_1$ is used to compute the conditional correlation. Comparing Fig.~\ref{fig:pre} and \ref{fig:post}, clearly, the absolute conditional corrections of outage branches change to zero after outages. The diagonal terms are the self-correlation and equal to one. This observation indicates that this proposed outage localization method is sensitive to outages and validates our proof in Theorem~\ref{thm:cond_indept}. When $f$ is unknown, by comparing Fig.~\ref{fig:pre} and \ref{fig:post-unknown}, we can still identify the outage lines. Therefore, the proposed method can still localize the out-of-service branches as accurate as the optimal approach.

\begin{figure}[h!]
    \centering
  \subfloat[Pre-outage\label{fig:pre}]{%
       \includegraphics[width=0.49\linewidth]{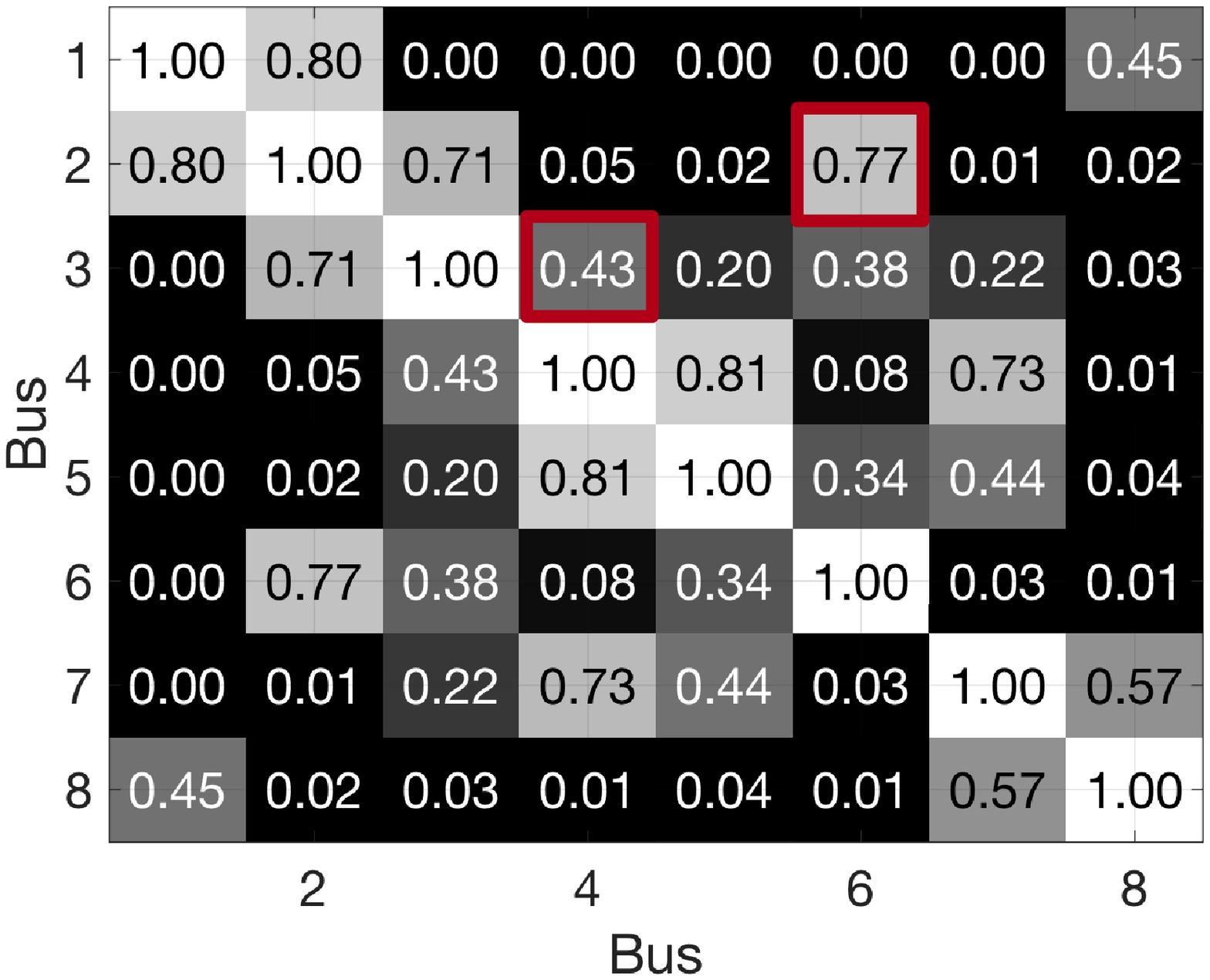}}
    \hfill
  \subfloat[Post-outage\label{fig:post}]{%
        \includegraphics[width=0.49\linewidth]{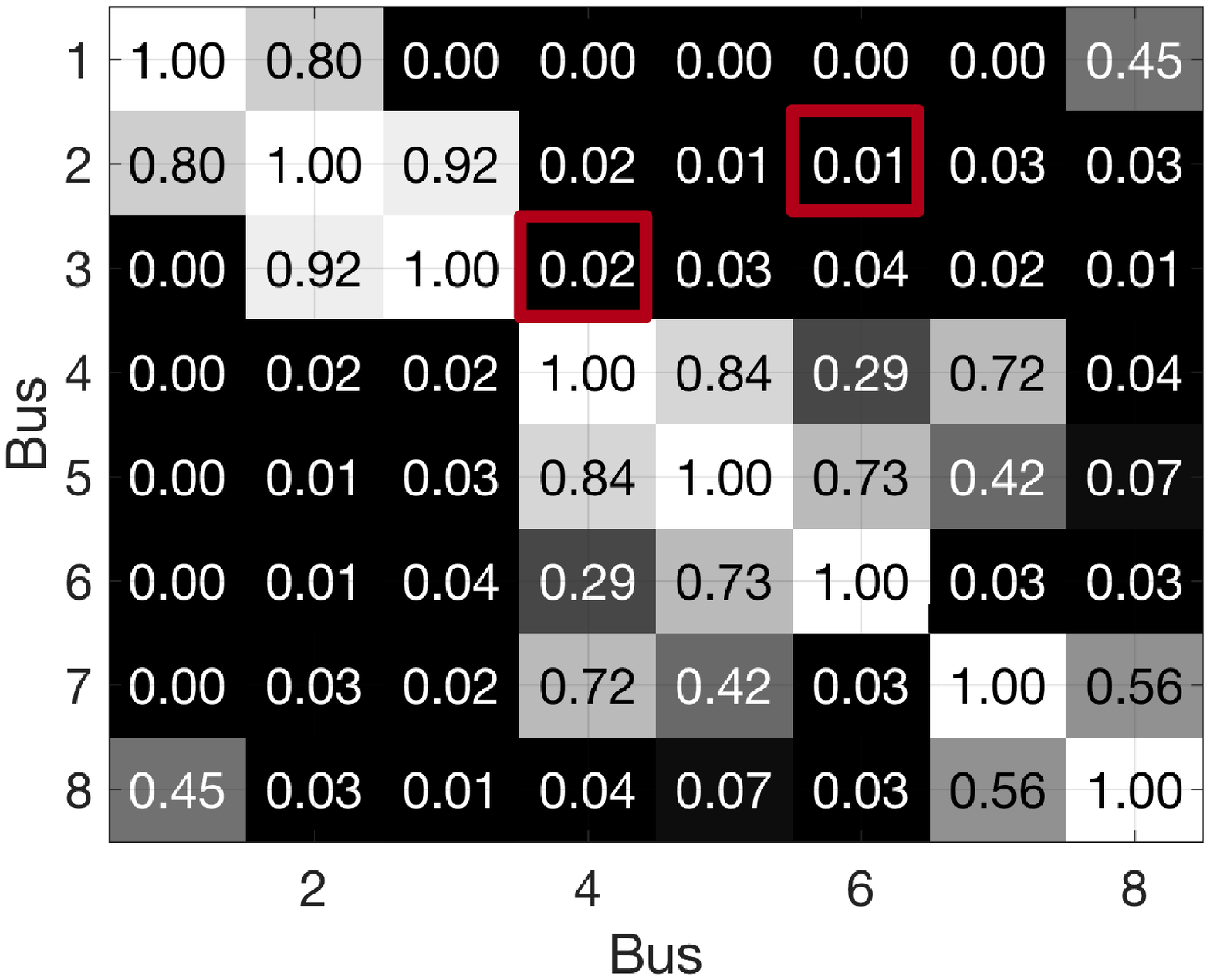}}
\\
  \subfloat[Post-outage with unknown distribution\label{fig:post-unknown}]{%
        \includegraphics[width=0.49\linewidth]{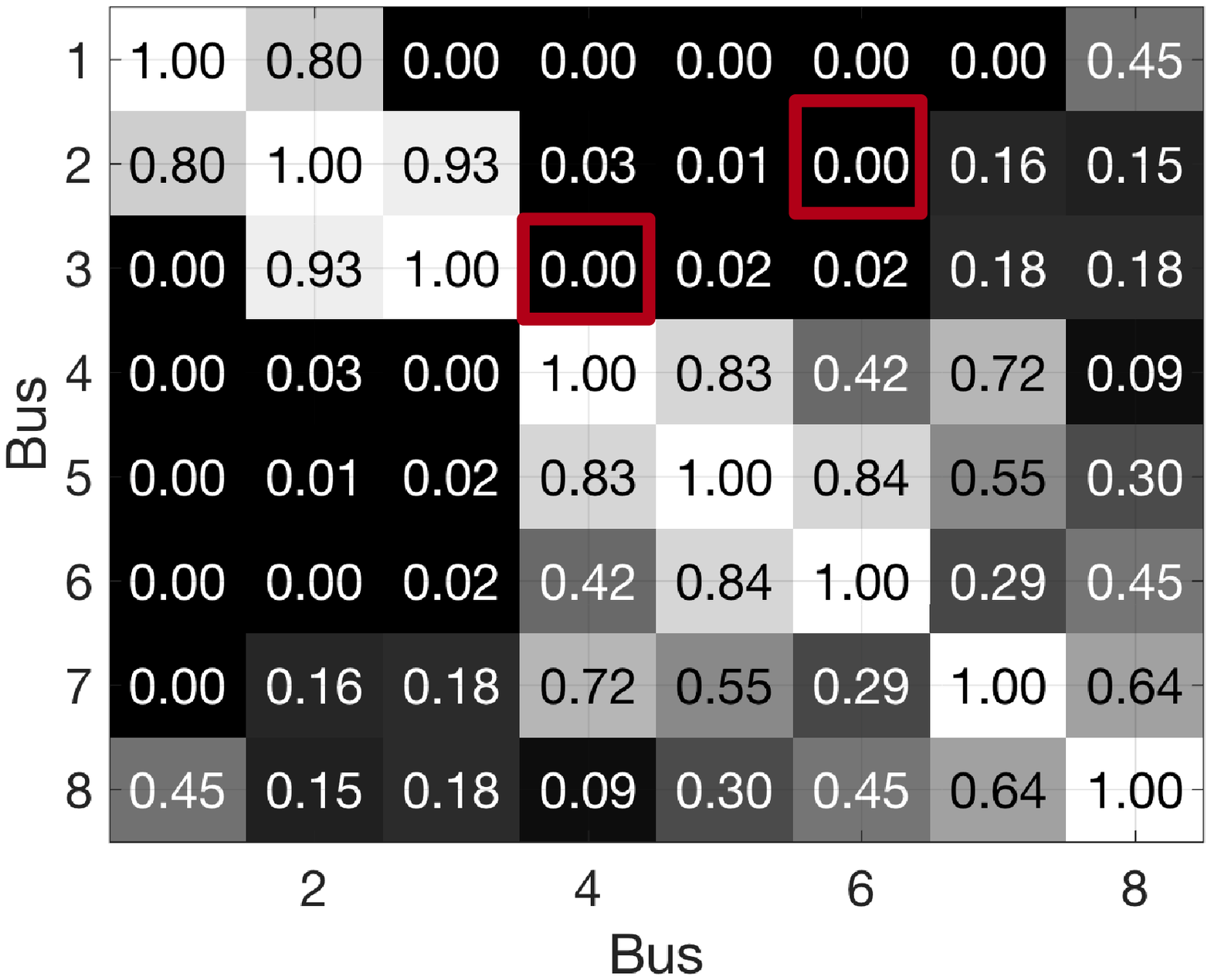}}
  \caption{Absolute conditional correlation of $8$-bus system before (a) and after (b \& c) an outage (Branches 3-4 and 2-6).}
    \label{fig:outage_line} 
\end{figure}


\subsection{Sensitivity to Data Resolutions} 
The ``ADRES-Concept'' project load profile \cite{Einfalt11, VUT16} is used to understand the proposed approach's sensitivity to data resolution. This data set contains real and reactive power profiles of 30 houses in Upper-Austria. The data were sampled every second over 14 days. The voltage data are generated using a subset of \textit{LV\_suburban\_mesh} grid with 33 DERs integrated. To simulate the damage patterns, we randomly set two branches to be out-of-service. Fig.~\ref{fig:data_resolution} shows the average detection delay with different data resolutions. The results are produced using Monte Carlo simulation over $1,000$ replications. We can see that with the increase of data resolution, the average detection delay is decreasing. The reason is that the distribution change is more significant when data resolution is large. However, for the absolute detection time delay, the high resolution data sources require less time. For example, only 5 seconds are needed to detect outages when the sampling rate is 1 second. As the sampling frequency is reduced, less data samples are required for detection. When the sampling rate is 1 minute, the proposed algorithm needs less than 3 minutes to detect outages. When the sampling rate is 30 minutes and one hour, the average detection delay is zero. \rev{Therefore, the major bottleneck of the detection delay is the sampling period. As discussed in Section~\ref{sec:outage_identify}, the computational complexity of both outage detection and localization only depends on the grid size. The computational time of a grid with less than $200$ buses is within $10$ seconds. Most smart meter systems have a sampling rate between 1 minute and 1 hour today. Hence, our algorithm can immediately detect the outage when the post-outage measurement is available.}

\putFig{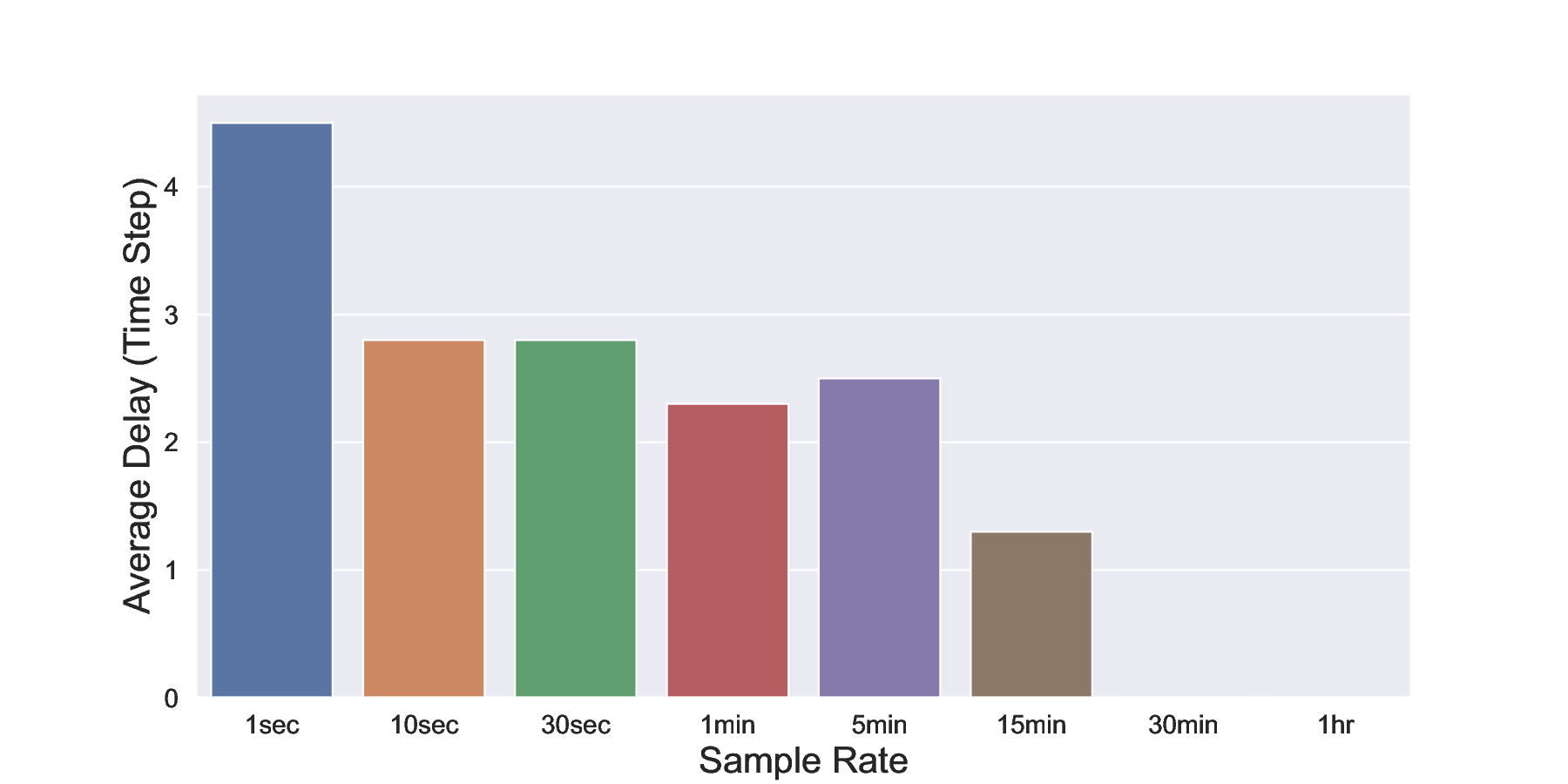}{The average detection delay with different data resolutions. $\alpha=10^{-5}$.}{\linewidth}

\subsection{Sensitivity to Data Accuracy}
Smart meter measurements are usually noisy. Thus, the analysis of our algorithm under different levels of measurement noises is critical for understanding the performance in a real-world scenario. In the U.S., ANSI C12.20 standard (Class 0.5) permits the utility smart meters to have an error within $\pm 0.5\%$ \cite{zheng2013smart, ansc12}. The standards in other countries have the similar requirement, e.g.,\cite{stategridmeter}. Table~\ref{tab:noise} shows the average detection delay with different noise levels over 1,000 iterations. The simulation setup is identical to the one in Section~\ref{sec:mesh_sim} and \ref{sec:radial_sim}. Hence, the data resolution is one hour. When noise level is less than $0.1\%$, the detection delay is similar to the detection delay of noiseless measurements. Since most measurement noises are zero-mean additive noise and we use the measurements' sufficient statistics for outage detection, the noise only impacts the estimation of covariance matrix. When noise level is $0.2\%$, one more data point is needed for detection.

\begin{table}[h!]
	\caption{Average Detection Delay (Time Step) of Line Outage Detection with DERs under Different Noise Levels. $\alpha = 10^{-5}$. The Post-outage Distribution $f$ is Unknown. Only Voltage Magnitudes $\Delta |\mathbf{V}_\mathcal{S}|$ are Used.}
	\centering
	\begin{tabular}{|c||c|c|}
	\hline
	Noise level & \textit{LV\_suburban\_mesh} & \textit{MV\_rural} \\
	\hline
	$0\%$ & 5.83 & 1.29 \\
	\hline
	$0.05\%$ & 5.42 & 1.32 \\
	\hline
	$0.1\%$ & 6.22 & 1.83  \\
	\hline
	$0.2\%$ & 7.90 & 2.53 \\
	\hline
	\end{tabular}
	\label{tab:noise}
\end{table}

Other types of device malfunctions may also impact the algorithm performances. For example, if the smart meter is not well calibrated, it may consistently produce measurements that are lack of precision and accuracy. For the proposed method, rather than directly use the raw measurements, we use sufficient statistics of data, e.g., mean and variance, for outage detection. Thus, if the systematic error persists across all measurements, our algorithm can still report outages. The anomaly data can degrade the performance of our proposed algorithm because the underlying data statistics may change due to non-outage events. There are multiple ways to minimize impacts. For example, data cleansing can be applied before processing data for outage identification. Also, we can reduce the probability of false alarm $\alpha$ to increase the confidence of outage report.

\subsection{Short-Circuit Faults Identification}
\label{sec:short_circuit}
\rev{
Besides line outages, another category of fault in distribution grids is short-circuit faults. In this subsection, we simulate the short-circuit fault in the radial IEEE 123-bus system via CYME power system analysis software. The fault scenarios are summarized in Table~\ref{tab:short_circuit}. In the short-circuit fault analysis, the data collected before $n=20$ are pre-fault and the measurements collected after $n=20$ are post-fault. As plotted in Fig.~\ref{fig:short_circuit_fault}, the means of voltage magnitudes $|V[n]|$ change significantly after faults. For fault type LG, LL, and LLG, the nodal nominal voltages drop to a non-zero value. These fault types are similar to the outage scenario 1, which we have discussed earlier in this section. As shown in Table.~\ref{tab:short_circuit}, less than one time step is needed to detect these faults. As the reference, to detect line outage between bus 67 and bus 160, the proposed algorithm has zero detection delay. For LLL and LLLG, the nodal nominal voltages drop to zero. Thus, the detection case is similar to the outage scenario 3. Our algorithm can immediately detect faults after they occurred. 
}
\begin{table}[h!]
    \caption{Short-Circuit Fault Types in the radial IEEE 123-bus System}
    \label{tab:short_circuit}
    \centering
    \begin{tabular}{|c|c|c|}
    \hline
   Branch & Fault Type & Detection Delay\\
   &&  (Time Step) $D(\tau)$ \\
    \hline
    67-160 & Single line-to-ground fault (LG) & 0.9\\
    & on phase A & \\
    \hline
    67-160 & Line-to-line fault (LL) & 0.5\\
    & on phase AB & \\
    \hline
    67-160 & Double Line-to-ground fault (LLG)  & 0\\
    & on phase AB & \\
    \hline
    67-160 & Three-phase short-circuit fault (LLL) & 0 \\
    \hline
    67-160 & Three-phase-to-ground fault (LLLG) & 0\\
    \hline
    \end{tabular}
\end{table}

\begin{figure}[h!]
    \centering
    \subfloat[Bus 67\label{fig:bus67}]{
    \includegraphics[width=0.98\linewidth]{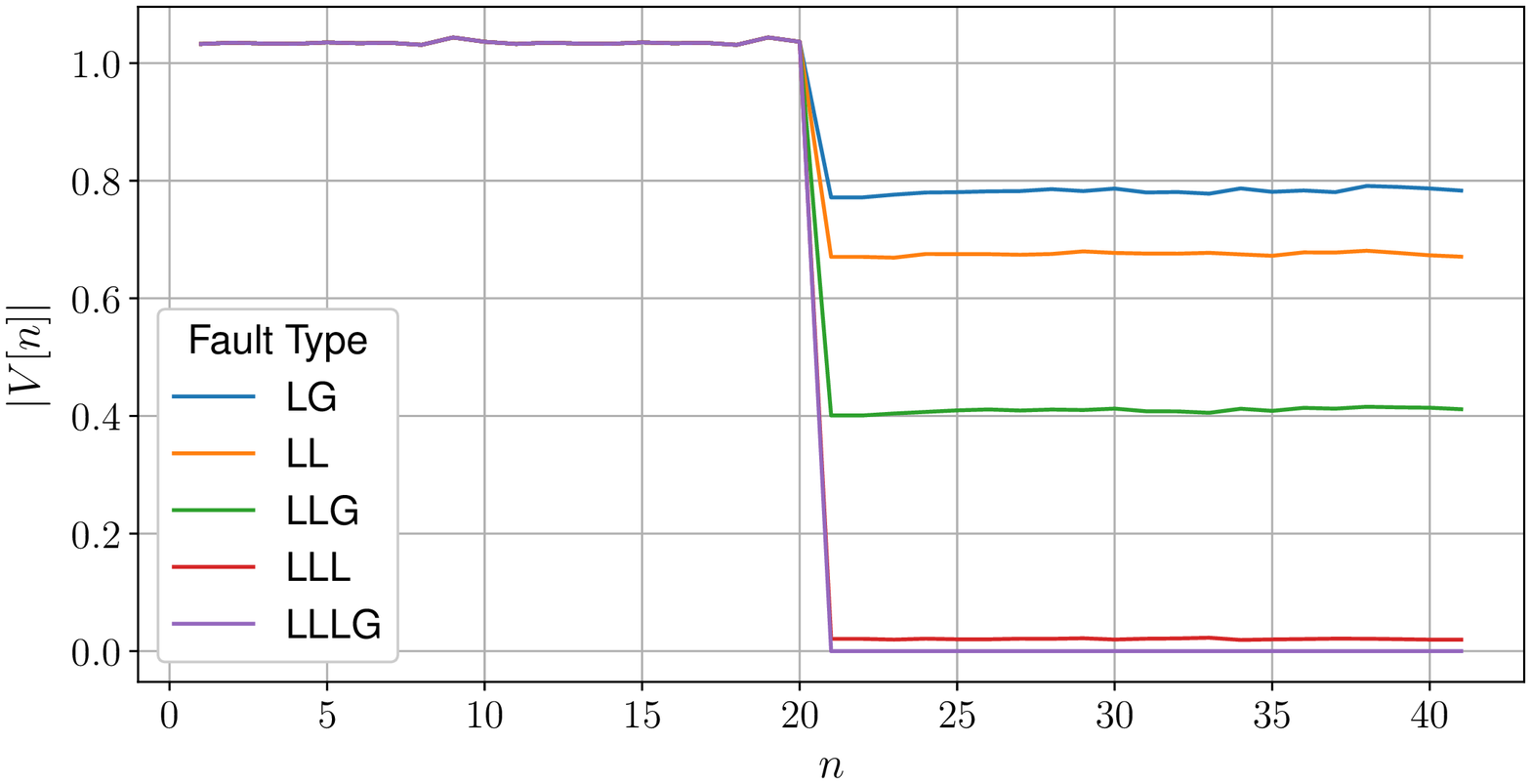}
    }
    \hfill
    \subfloat[Bus 160\label{fig:bus160}]{
    \includegraphics[width=0.98\linewidth]{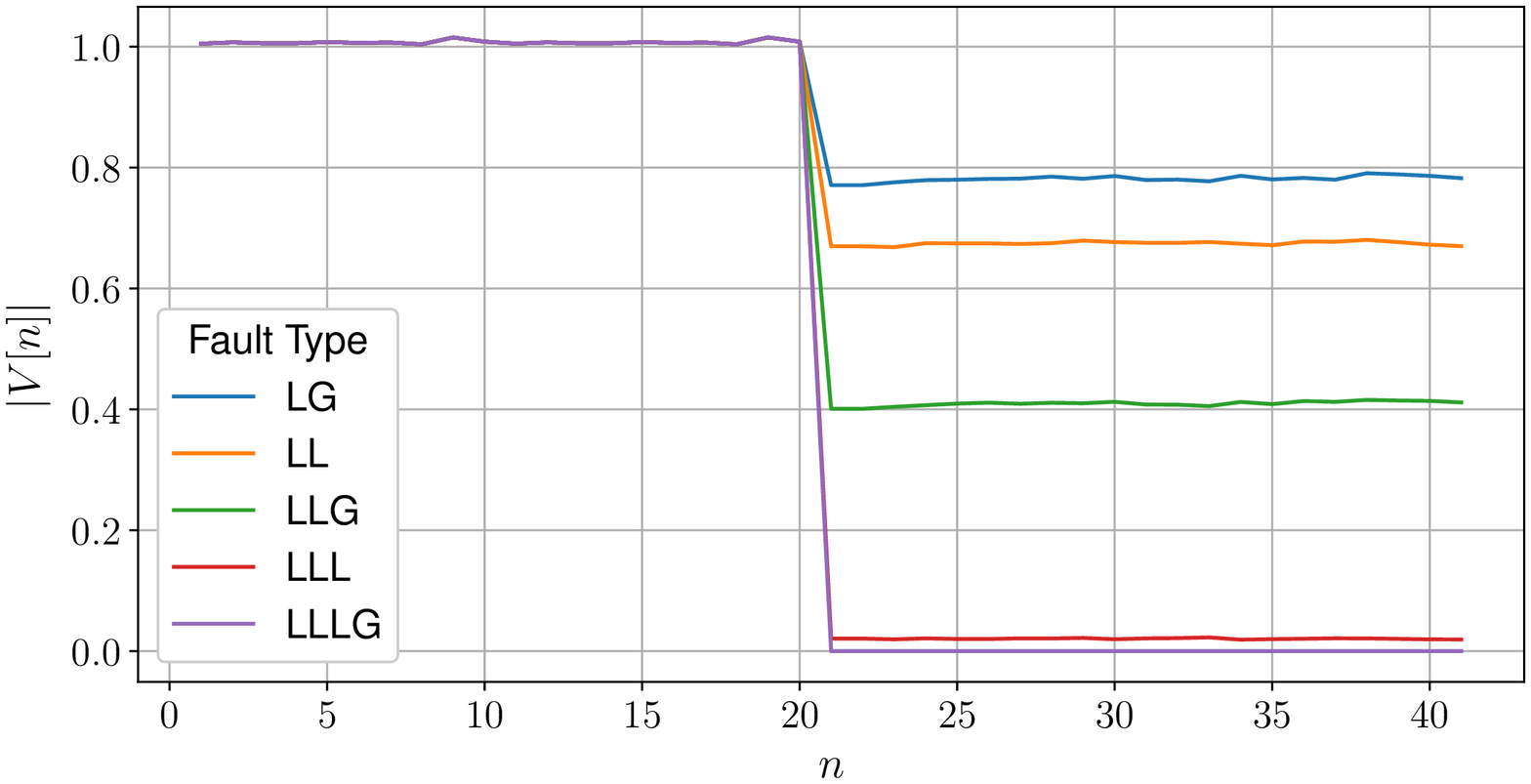}
    }
    \caption{\rev{Nodal nominal voltage measurements in per unit for different fault types in the radial IEEE 123-bus system.}}
    \label{fig:short_circuit_fault}
\end{figure}

\section{Conclusion}\label{sec:con}
 In this paper, we propose a new approach to automatically detect and identify outages in urban distribution grids with high renewable penetration. Specifically, we develop a stochastic modeling of nodal voltage data stream and propose a change point detection approach based on the probability distribution changes due to outage events. As a highlight, unlike existing approaches, our method is applicable to the existing distribution grids because we require neither the grid topology nor the outage pattern as a prior. Also, we only need smart data measurements to achieve the optimal detection performance. In addition to outage detection, we provide theoretical proof that optimal out-of-service branch identification can be achieved due to the conditional independence of voltages based on the power flow analysis. We verify the proposed algorithm on eight mesh and radial distribution grid systems with and without DERs. From extensive simulations, our algorithm can perfectly detect and identify outages in a short time, with and without the integration of DERs.

There are multiple future works that can further enhance the proposed algorithm. For example, as discussed in Section~\ref{sec:radial_sim}, the computationally complexity of line outage detection scales up with the growth of grid size. A distributed or decentralized approach may reduce the computational efforts. Additionally, we discuss the impact of measurement noise, systemic errors, and anomaly data in this paper. However, other types of data quality issues may also degrade the performance, such as missing data and fixed-point measurements. How to handle these practical scenarios requires further investigation. In this paper, we focus on detecting and localizing the line outage event. Identifying fault types that cause line outages is also an interesting research direction. At last, dynamic topology estimation and switch status identification share some similarities with the proposed out-of-branch localization method, but also have more rigorous requirements. How to apply the proposed out-of-branch localization method to estimate topology is a direction of future studies.

\section{Acknowledgement}
We would like to thank Jingyi Yuan from Arizona State University for discussion on the short-circuit fault analysis.

\appendix
\subsection{Proof of Lemma~\ref{lemma:mle}}
\label{sec:parm_est}
\begin{proof}
To apply the maximum likelihood method, we need to compute the partial derivative of the posterior probability $P(\mathcal{H}_1|\Delta\mathbf{v}^{1:N})$. Unfortunately, $P(\mathcal{H}_1|\Delta\mathbf{v}^{1:N})$ is not a convex function and we may have multiple estimates. To address this challenge, we will provide an approximation of the posterior probability $P(\mathcal{H}_1|\Delta\mathbf{v}^{1:N})$. Specifically, the log-probability $\log P(\mathcal{H}_1|\Delta\mathbf{v}^{1:N})$ is
\begin{align}
&	\log P(\mathcal{H}_1|\Delta\mathbf{v}^{1:N}) \nonumber \\
=& \log C + \log\left\{\sum_{k=1}^N\pi(k)\prod_{n=1}^{k-1}g(\mv[n])\prod_{n=k}^{N}f(\mv[n]; \boldsymbol{\Theta})\right\}, \label{eq:log_post}
\end{align}
where $\boldsymbol{\Theta} = \{\mu_1, \Sigma_1\}$ represents the unknown parameters of $f$. In (\ref{eq:log_post}), the term within the braces is an expectation of $\prod_{n=1}^{k-1}g(\mv[n])\prod_{n=k}^{N}f(\mv[n]; \boldsymbol{\Theta})$ over the prior distribution $\pi$, $\E_\pi(\prod_{n=1}^{k-1}g(\mv[n])\prod_{n=k}^{N}f(\mv[n]; \boldsymbol{\Theta}))$. Also, the logarithmic function is convex. Therefore, we can apply the Jensen's inequality \cite{cover2012elements} to approximate $\log P(\mathcal{H}_1|\Delta\mathbf{v}^{1:N})$:
\begin{align}
	& \log P(\mathcal{H}_1|\Delta\mathbf{v}^{1:N}) \nonumber \\
	\geq & \log C + \sum_{k=1}^N\pi(k)\left(\sum_{n=1}^{k-1}\log g(\mv[n]) + \sum_{n=k}^{N} \log f(\mv[n]; \boldsymbol{\Theta})\right) \nonumber \\
	= & \widetilde{P}(\mathcal{H}_1|\Delta\mathbf{v}^{1:N}) \label{eq:approx_log_post}.
\end{align}
Since $g$ and $f$ are Gaussian distributions, (\ref{eq:approx_log_post}) can be written as
\begin{align*}
    & \widetilde{P}(\mathcal{H}_1|\mv^{1:N}) 
    = \log C + \sum_{k=1}^N\frac{-\pi(k)}{2} \cdot \\
    &\left(
    \sum_{n=1}^{k-1}\log|2\pi\Sigma_0| + (\mv[n] - \mu_0)^T\Sigma_0^{-1}(\mv[n] - \mu_0) \right. \\
    &+ \left. \sum_{n=k}^{N}\log|2\pi\Sigma_1| + (\mv[n] - \mu_1)^T\Sigma_1^{-1}(\mv[n] - \mu_1)
    \right).
\end{align*}
Since $\widetilde{P}(\mathcal{H}_1|\Delta\mathbf{v}^{1:N})$ is convex, we can estimate $\mu_1$ by setting $\partial \widetilde{P}/\partial \mu_1=0$. Specifically, we have
\[
\frac{\partial \widetilde{P}(\mathcal{H}_1|\mv^{1:N})}{\partial \mu_1} = \sum_{k=1}^N \frac{-\pi(k)}{2}\sum_{n=k}^N(\mv[n]-\mu_1)\Sigma_1^{-1} = 0.
\]
Since
\[
    \sum_{n=k}^N(\mv[n]-\mu_1) = \left(\sum_{n=k}^N\mv[n] - (N-k+1)\mu_1\right),
\]
the estimate of $\mu_1$ is
\[
\widehat{\mu}_1 = \frac{\sum_{k=1}^N\pi(k)\sum_{n=k}^N\mv[n]}{\sum_{k=1}^N\pi(k)(N-k+1)}.
\]

For the covariance matrix $\Sigma_1$, the partial derivative is 
\[
   \frac{\partial \widetilde{P}(\mathcal{H}_1|\mv^{1:N})}{\partial \Sigma_1} = \sum_{k=1}^N\frac{-\pi(k)}{2}\left(\sum_{n=k}^N S[k] - (N-k+1)\Sigma_1\right)
\]
where $S[k] = \sum_{n=k}^N(\mv[n] - \mu_1)(\mv[n] - \mu_1)^T$. Letting $\mu_1 = \widehat{\mu}_1$ and $\partial \widetilde{P}(\mathcal{H}_1|\mv^{1:N})/\partial \Sigma_1 = 0$, the covariance matrix estimate is 
\[
\widehat{\Sigma}_1 = \frac{\sum_{k=1}^N\pi(k)S[k]}{\sum_{k=1}^N\pi(k)(N-k+1)}.
\]
\end{proof}

\bibliographystyle{IEEEtran}
\bibliography{ref}
\end{document}